%% file: pats_to_jcss.tex
\documentclass[a4paper,11pt,titlepage]{article}
\pdfoutput=1

\usepackage{fixltx2e}
\usepackage[T1]{fontenc}
\usepackage{lmodern}
\usepackage{microtype}
\usepackage{geometry}
\usepackage[UKenglish]{babel}
\usepackage[numbers,sort]{natbib}
\usepackage{amsmath}
\usepackage{amssymb}
\usepackage{amsthm}
\usepackage[all,warning]{onlyamsmath}
\usepackage{enumerate}
\usepackage{graphicx}
\usepackage{latexsym}
\usepackage{mathrsfs}
\usepackage{subfigure}
\usepackage{url}
\usepackage{paralist}
\usepackage{mathtools}
\usepackage[font=small]{caption}
\usepackage[hidelinks]{hyperref}

\hypersetup{
  pdftitle = {Search Methods for Tile Sets in Patterned DNA Self-Assembly},
  pdfauthor = {Mika G\"{o}\"{o}s, Tuomo Lempi\"{a}inen, Eugen Czeizler, Pekka Orponen}
}

\pdfpageattr{/Group << /S /Transparency /I true /CS /DeviceRGB>>}

\newcommand{\Z}{\ensuremath{\mathbb{Z}}}
\newcommand{\Prod}{\ensuremath{\operatorname{Prod}}}
\newcommand{\Term}{\ensuremath{\operatorname{Term}}}

\theoremstyle{plain}
\newtheorem{theorem}{Theorem}
\newtheorem{lemma}[theorem]{Lemma}
\newtheorem{proposition}[theorem]{Proposition}
\newtheorem{corollary}[theorem]{Corollary}

\theoremstyle{definition}
\newtheorem{definition}[theorem]{Definition}

\begin{document}

\begin{titlepage}
  \setcounter{page}{0}
  \begin{center}
    \vspace*{5ex}
    {\LARGE \textbf{Search Methods for Tile Sets in Patterned\\ DNA Self-Assembly}\par}
    \bigskip
    \bigskip
    Mika G\"{o}\"{o}s$^1$,
    Tuomo Lempi\"{a}inen$^2$,
    Eugen Czeizler, and
    Pekka Orponen\\
    \bigskip
    \small{Department of Information and Computer Science and\\
    Helsinki Institute for Information Technology HIIT\\
    Aalto University, Finland}
  \end{center}
  \bigskip
  \bigskip
  \noindent\textbf{Abstract.}
  \input{abstract}
  \bigskip
  \bigskip
  \bigskip\\
  \small{This is the author's version of a work that was accepted for publication in \emph{Journal of Computer and System Sciences}. Changes resulting from the publishing process, such as peer review, editing, corrections, structural formatting, and other quality control mechanisms may not be reflected in this document. Changes may have been made to this work since it was submitted for publication. A definitive version was subsequently published in Journal of Computer and System Sciences, volume 80, issue 1, pages 297--319, February 2014, \href{http://dx.doi.org/10.1016/j.jcss.2013.08.003}{doi:10.1016/j.jcss.2013.08.003}.}
  \bigskip\\
  \small{Preliminary versions of parts of this work have appeared in the \emph{Proceedings of the 16th} and \emph{17th International Conference on DNA Computing and Molecular Programming} (Hong Kong, China, June 2010 and Pasadena, CA, USA, September 2011, respectively)~\cite{GoOr11,LeCO11}.}
  \bigskip\\
  Current affiliations:\\[1ex]
  \small{$^1$Department of Computer Science, University of Toronto}\\[1ex]
  \small{$^2$Helsinki Institute for Information Technology HIIT and Department of Computer Science, University of Helsinki}
\end{titlepage}

\section{Introduction}
\label{sec:intro}
\input{intro}

\section{Preliminaries}
\label{sec:prelim}
\input{prelim}

\section{The Search Space of Consistent Tile Sets}
\label{sec:search_space}
\input{search_space}

\section{Complete Search for Minimal Tile Sets}
\label{sec:complete_search}
\input{complete_search}

\section{Heuristically Guided Search for Small Tile Sets}
\label{sec:local_search}
\input{local_search}

\section{Answer Set Programming for Minimal Tile Sets}
\label{sec:asp}
\input{asp}

\section{The Reliability of Tile Sets}
\label{sec:reliability}
\input{reliability}

\section{Conclusions}
\label{sec:conclusions}
\input{conclusions}

\section*{Acknowledgements}
\input{acks}

\bibliographystyle{plainnat}
{\small\bibliography{sources}}

\newpage
\appendix

\section{A Design Framework for Carbon Nanotube Circuits Affixed on DNA Origami Tiles}
\label{sec:cnfet}
\input{cnfet}

\end{document}

%% file: abstract.tex
The Pattern self-Assembly Tile set Synthesis (PATS) problem, which
arises in the theory of structured DNA self-assembly, is to
determine a set of coloured tiles that, starting from a bordering seed
structure, self-assembles to a given rectangular colour pattern. The
task of finding minimum-size tile sets is known to be \textsf{NP}-hard. We
explore several complete and incomplete search techniques for finding
minimal, or at least small, tile sets and also assess the reliability
of the solutions obtained according to the kinetic Tile Assembly
Model.

%% file: intro.tex
Algorithmic assembly of nucleic acids (DNA and RNA) has advanced
extensively in the past 30 years, from a seminal idea to the current
designs and experimental implementations of complex nanostructures and
nanodevices with dynamic, programmable evolution and machinelike
properties. Recent developments in the field include fundamental
constructions such as \textit{in vitro} complex 3D pattern formation and
functionalisation \cite{DDLH09,KuLT09}, robotic designs such as mobile
arms, walkers, motors \cite{LMDM10,ZOKV11}, computational
primitives~\cite{QiWi11,QiWB11}, and also applications to \textit{in
vivo} biosensors \cite{LiCL09} and potential drug delivery mechanisms
and therapeutics \cite{LPZL11}.

Self-assembly of nucleic acids can be seen both as a form of
structural nanotechnology and as a model of computation. As a
computational model, one first encodes the input of a computational
problem into an algorithmically designed (DNA) pattern or shape. Then,
by making use of both the initial oligomer design and the intrinsic
properties of the self-assembly system, one manipulates the structure
to produce a new architecture that encodes the desired output.

As a nanotechnology, the goal of algorithmic (DNA/RNA) self-assembly
is to design oligomer sequences that in solution would autonomously
(or with as little interaction as possible) assemble into complex
polymer structures. These may have both static and dynamic properties,
may bind other molecules such as gold nanoparticles or various
proteins, may act as fully addressable scaffolds, or may be used for
further manipulation.  Such molecular constructions can be composed of
from only a couple of DNA strands to more than 200 and, in some cases,
can change their conformation and achieve distinct functionalities.

In recent years there has been a growing interest in integrating these
two directions, in order to obtain complex supramolecular
constructions with interdependencies between computational functions
and conformational switching. Such approaches are envisioned due to a
key property of nucleic acid scaffolds, viz.\ their modularity:
multiple functional units can be attached to a common scaffold, thus
giving rise to multifunctional devices. Thus, the self-assembly of
nanostructures templated on synthetic DNA has been proposed by several
authors as a potentially ground-breaking technology for the
manufacture of next-generation circuits, devices and
materials~\cite{KSML10,MHBB10,WLWS98,YPFR03}.  Also laboratory
techniques for synthesising the requisite 2D DNA template lattices,
many based on Rothemund's~\cite{Roth06} DNA origami tiles, have
recently been demonstrated by many groups~\cite{LZWS11,REKH11}.

In order to support the manufacture of aperiodic structures, such as
electronic circuit designs, these DNA templates need to be
addressable. When the template is constructed as a tiling from a
family of DNA origami (or other kinds of) tiles, one can view the base
tiles as being ``coloured'' according to their different
functionalities, and the completed template implementing a desired
colour pattern.\footnote{For examples of such tile-based high-level
designs for nano-electric circuits cf.\ Appendix~\ref{sec:cnfet}, which summarises
a scheme from Czeizler et al.~\cite{CzLO11}.}  Now, a given target
pattern can be assembled from many different families of base tiles,
and to improve the laboratory synthesis it is advantageous to try to
minimise the number of tile types needed and/or maximise the
probability that they self-assemble to the desired pattern, given some
characteristics of tiling errors.

The task of minimising the number of DNA tile types required to
implement a given 2D pattern was identified by Ma and
Lombardi~\cite{MaLo08}, who formulated it as a combinatorial
optimisation problem, the \emph{Pattern self-Assembly Tile set
Synthesis} (PATS) problem, and also proposed two greedy heuristic
algorithms for solving the task. The problem was recently proved to be
\textsf{NP}-hard~\cite{CzPo12,Seki13}, and hence finding an absolutely
minimum-size tile set for a given pattern most likely requires an
exponential amount of time in the worst case. Thus the problem needs
to be addressed either with complete methods yielding optimal tile
sets for small patterns, or incomplete methods that work also for
larger patterns but do not guarantee that the tile sets produced are
of minimal size. In this work, we present search algorithms covering
both approaches and assess their behaviour experimentally using both
randomly generated and benchmark pattern test sets.  We attend both to
the running time of the respective algorithms, and to the size and
assembly reliability of the tile sets produced.

In the following, we first in Section~\ref{sec:prelim} present an
overview of the underlying tile assembly model~\cite{Winf98b,RoWi00}
and the PATS problem~\cite{MaLo08}, and then in
Section~\ref{sec:search_space} discuss the search space of
pattern-consistent tile sets (viewed abstractly as partitions of the
ambient rectangular grid). In Section~\ref{sec:complete_search} we
proceed to describe our exhaustive partition-search branch-and-bound
algorithm (PS-BB) to find tile sets of absolutely minimum cardinality.
The algorithm makes use of a search tree in the lattice of grid
partitions, and an efficient bounding function to prune this search
tree.

While the PS-BB algorithm can be used to find certifiably minimal tile
sets for small patterns, the size of the search space grows so rapidly
that the algorithm hits a complexity barrier at approximately pattern
sizes of $7\times 7$ tiles, for random test patterns. Thus, in a
second approach, presented in Section~\ref{sec:local_search}, we
tailor the basic partition-search framework of the PS-BB algorithm
towards the goal of finding small, but not necessarily minimal tile
sets.  Instead of a systematic branch-and-bound pruning and traversal
of the complete search space, the modified algorithm PS-H applies
heuristics which attempt to optimise the order of the directions in
which the space is explored.

It is well known in the heuristic optimisation
community~\cite{GoSe01,LuSZ93} that when the runtime distribution of a
randomised search algorithm has a large variance, it is with high
probability more efficient to run several independent short runs
(``restarts'') of the algorithm than a single long run.
Correspondingly, we investigate the efficiency of the PS-H algorithm
for a number of parallel executions ranging from 1 to 32, and note
that indeed this number has a significant effect on the success rate
of the algorithm in finding small tile sets.

As a third alternative, presented in Section~\ref{sec:asp}, we
formulate the PATS problem as an Answer Set Programming (ASP)
task~\cite{Lifs08}, and apply a generic ASP solver to find solutions to
it. Here our experimental results indicate that for patterns with a
small optimal solution, the ASP approach indeed works well in
discovering that solution.

Given the inherently stochastic nature of the DNA self-assembly
process, it is important also to assess the reliability of a given
tile set, i.e.\ the probability of its error-free self-assembly to the
desired target pattern. In Section~\ref{sec:reliability} we introduce
a method for estimating this quantity, based on Winfree's analysis of
the kinetic Tile Assembly Model~\cite{Winf98b}. We present
experimental data on the reliability of tile sets found by the PS-BB
and PS-H algorithms and find that also here the heuristic
optimisations introduced in the PS-H approach result in a notable
improvement over the basic PS-BB method.

%% file: prelim.tex
In this section, we first briefly review the abstract Tile Assembly
Model (aTAM) as introduced by Winfree and
Rothemund~\cite{Winf98b,RoWi00} and then summarise the PATS
problem~\cite{MaLo08}.

\subsection{The Abstract Tile Assembly Model}

The aTAM is a custom-made generalisation of Wang tile systems~\cite{Wa61,Gr86}, designed for the study of
self-assembly systems. The basic components of the aTAM are
non-rotatable unit square tiles, uniquely defined by the sets of four
``glues'' assigned to their edges. The glues come from a finite
alphabet, and each pair of two glues is associated a strength value
that determines the stability of a link between two tiles having these
glues on the abutting edges. In most cases, it is assumed that the
strength of two distinct glues is zero, while a pair of matching glues
has strength either 1 or 2.

Let $\mathcal{D} = \{N,E,S,W\}$ be the set of four functions
$\Z^2\to\Z^2$ corresponding to the four cardinal
directions:\footnote{In many cases, we use the elements of
  $\mathcal{D}$ just as direction labels and do not interpret
  them as functions. However, in these cases too, we identify
  $S=N^{-1}$ and $W=E^{-1}$.}
$N(x,y)=(x,y+1)$, $E(x,y)=(x+1,y)$, $S=N^{-1}$ and $W=E^{-1}$.
Let $\Sigma$ be a finite set of \emph{glue types} and
$s\colon\Sigma\times\Sigma \to \mathbb{N}$ a \emph{glue strength}
function such that, unless otherwise specified,
$s(\sigma,\sigma') > 0$ only if $\sigma = \sigma'$.
A \emph{tile type} $t\in\Sigma^4$ is a quadruple
$(\sigma_N(t), \sigma_E(t), \sigma_S(t), \sigma_W(t))$ of
glue types for each side of the unit square.  A \emph{tile system}
$T\subseteq\Sigma^4$ is a finite collection of different tile types.

A \emph{(tile) assembly}~$\mathcal{A}$ is a partial mapping
$\mathcal{A}\colon\Z^2 \rightarrow \Sigma^4$ that assigns
tiles to locations in the two-dimensional grid.
A \emph{tile assembly system} (TAS)
$\mathscr{T}=(T,\mathcal{S},s,\tau)$ consists of a tile system~$T$, a
\emph{seed assembly}~$\mathcal{S}$, a glue strength function~$s$ and a
\emph{temperature} $\tau \in \Z^+$ (we use $\tau = 2$). The seed
structure~$\mathcal{S}$ can be either an individual tile or a
connected, finite assembly.  Given an existing (connected)
assembly~$\mathcal{A}$, such as the seed structure~$\mathcal{S}$, a
tile from $T$ can adjoin the assembly if the total strength of the
binding, given by the sum of all strength function values among the
glues placed on the boundary between the tile and the assembly,
reaches or surpasses the temperature threshold~$\tau$.  Note that
tiles of the seed assembly~$\mathcal{S}$ do not need to be in the tile
system $T$, but that $\mathcal{S}$ can be extended only by tiles from
$T$.

Formally, we say that assembly
$\mathcal{A}$ \emph{produces directly} assembly $\mathcal{A}'$,
denoted $\mathcal{A} \rightarrow_{\mathscr{T}} \mathcal{A}'$,
if there exists a site
$(x,y)\in\Z^2$ and a tile $t\in T$ such that $\mathcal{A}' =
\mathcal{A} \cup \{((x,y),t)\}$, where the union is disjoint, and
\begin{equation*} \sum_D
s(\sigma_D(t),\sigma_{D^{-1}}(\mathcal{A}(D(x,y))))\enspace\geq\enspace\tau,
\end{equation*} where $D$ ranges over those directions in $\mathcal{D}$
for which $\mathcal{A}(D(x,y))$ is defined.

In Figure~\ref{fig:bc} we present a TAS with seven tile types and
temperature $\tau=2$ which, starting from the seed tile, assembles a
continuously growing structure that corresponds to a binary counter
pattern (see Figure~\ref{fig:pat-bc32}). Out of the seven tile types
in Figure~\ref{fig:bc}(a), one can distinguish the tile~$\textbf{s}$
used as a seed, two tile types which assemble the boundary of the
structure, and four rule-tile types (two of which are distinguished by
$\textbf{x}$ and $\textbf{y}$), which fill the area in between the
L-shaped boundary. Considering the partial assembly presented in
Figure~\ref{fig:bc}(b), a tile of type~$\textbf{y}$ can adjoin the
assembly at position~$(4,3)$ since
\[
  s(\sigma_S(\mathbf{y}),\sigma_{S^{-1}}(\mathcal{A}(S(4,3)))) + s(\sigma_W(\mathbf{y}),\sigma_{W^{-1}}(\mathcal{A}(W(4,3)))) = 1+1 \geq \tau,
\]
while a tile of type~$\textbf{x}$ cannot adjoin the assembly at the same position (i.e.~$(4,3)$) since
\[
  s(\sigma_S(\mathbf{x}),\sigma_{S^{-1}}(\mathcal{A}(S(4,3)))) + s(\sigma_W(\mathbf{x}),\sigma_{W^{-1}}(\mathcal{A}(W(4,3)))) = 0+1 < \tau.
\]

\begin{figure}[t]
  \centering
  \includegraphics[width=\textwidth]{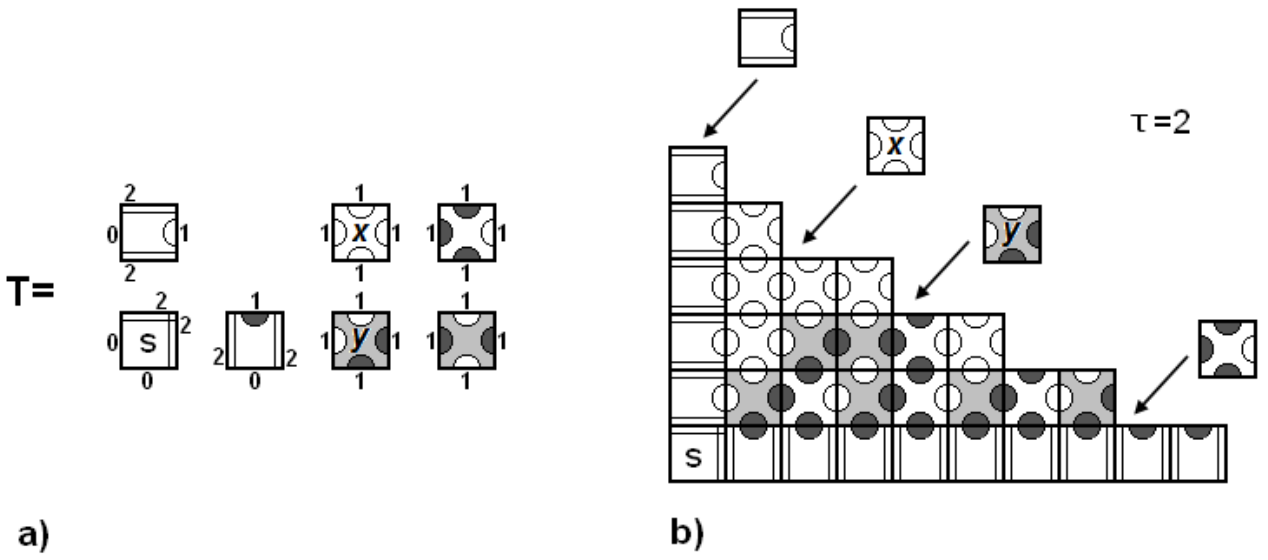}
  \caption{(a) The binary counter tile set~\cite{RoWi00}. The different glues are graphically differentiated, while their associated strengths are marked numerically. The colour of the tiles is an indicator of which tiles represent a black spot and which tiles represent a white spot in the pattern. (b) The assembly of the binary counter pattern for a TAS using the tile set~$T$, a seed structure consisting of the single tile~$\textbf{s}$ and the temperature threshold~$\tau=2$.}
  \label{fig:bc}
\end{figure}

Let $\rightarrow^*_{\mathscr{T}}$ be the reflexive transitive closure
of $\rightarrow_{\mathscr{T}}$.  A TAS $\mathscr{T}$ \emph{produces}
an assembly~$\mathcal{A}$ if
$\mathcal{A}$ is an extension of the seed assembly $\mathcal{S}$, that is,
$\mathcal{S}\rightarrow^*_{\mathscr{T}}\mathcal{A}$. Denote by
$\Prod \mathscr{T}$ the set of all assemblies produced by
$\mathscr{T}$.
A TAS $\mathscr{T}$ is \emph{deterministic} if for any assembly $\mathcal{A}
\in \Prod \mathscr{T}$ and for every $(x,y)\in\Z^2$ there exists at
most one $t\in T$ such that $\mathcal{A}$ can be extended with $t$ at
site $(x,y)$.
Then the pair $(\Prod \mathscr{T}, \rightarrow^*_{\mathscr{T}})$
forms a partially ordered set, which is a lattice if and only if
$\mathscr{T}$ is deterministic. The maximal elements in
$\Prod \mathscr{T}$, i.e.\ the assemblies $\mathcal{A}$ for which
there does not exist any $\mathcal{A}'$ satisfying
$\mathcal{A}\rightarrow_{\mathscr{T}}\mathcal{A}'$, are called
\emph{terminal assemblies}.
Denote by $\Term \mathscr{T}$ the set of terminal assemblies
of $\mathscr{T}$.
In case of finite assemblies, an equivalent definition of determinism is that all \emph{assembly sequences} $ \mathcal{S} \rightarrow_{\mathscr{T}} \mathcal{A}_1
\rightarrow_{\mathscr{T}} \mathcal{A}_2 \rightarrow_{\mathscr{T}} \cdots $ terminate and $\Term \mathscr{T} = \{\mathcal{P}\}$ for
some assembly $\mathcal{P}$. In this case we say that $\mathscr{T}$ \emph{uniquely produces} $\mathcal{P}$.

\subsection{The PATS Problem}

Let the dimensions $m$ and $n$ be fixed. A mapping from
$[m]\times[n] \subseteq \Z^2$ onto $[k]$ defines
a \emph{$k$-colouring} or a \emph{$k$-coloured pattern}.
To build a given pattern, we start with boundary tiles in
place for the west and south borders of the $m$ by $n$ rectangle and
keep extending this assembly by tiles with strength-1 glues.

\begin{definition}[Pattern self-Assembly Tile set Synthesis
  (PATS) \cite{MaLo08}]\mbox{}\\[1ex]
  \begin{tabular}[]{rp{0.85\textwidth}}
    \textbf{Given:} & A $k$-colouring $c\colon [m]\times[n]\to[k]$.\\
    \textbf{Find:} & A tile assembly system $\mathscr{T} = (T,\mathcal{S},s,2)$ such that
    \vspace{1ex}
    \begin{compactenum}[P1. ]
      \item The tiles in $T$ have glue strength 1.
      \item The domain of $\mathcal{S}$ is $[0,m]\times\{0\}\cup\{0\}\times[0,n]$
        and all the terminal assemblies have domain $[0,m]\times[0,n]$.
      \item There exists a tile colouring $d\colon T\to[k]$ such that each
        terminal assembly $\mathcal{A}\in\Term\mathscr{T}$ satisfies
        $d(\mathcal{A}(x,y))=c(x,y)$ for all $(x,y)\in[m]\times[n]$.
    \end{compactenum}
  \end{tabular}
\end{definition}

Finding minimal solutions (in terms of $|T|$) to the PATS problem was
claimed to be \textsf{NP}-hard by Ma and Lombardi~\cite{MaLo08} and
proved to be so by Czeizler and Popa~\cite{CzPo12}.\footnote{An
  improvement of the result to use only a constant number of tile
  colours is due to Seki~\cite{Seki13}.}  Without loss of generality,
we consider only TASs $\mathscr{T}$ in which every tile type
participates in some terminal assembly of $\mathscr{T}$.

As an illustration, using a 4-tile TAS from Winfree~\cite{Winf98b}, we construct a $7
\times 7$ Sierpinski triangle pattern in Figure~\ref{fig:sierpinski}. We use
natural numbers as glue labels in our figures.

\begin{figure}[t] \centering
  \includegraphics[width=\textwidth]{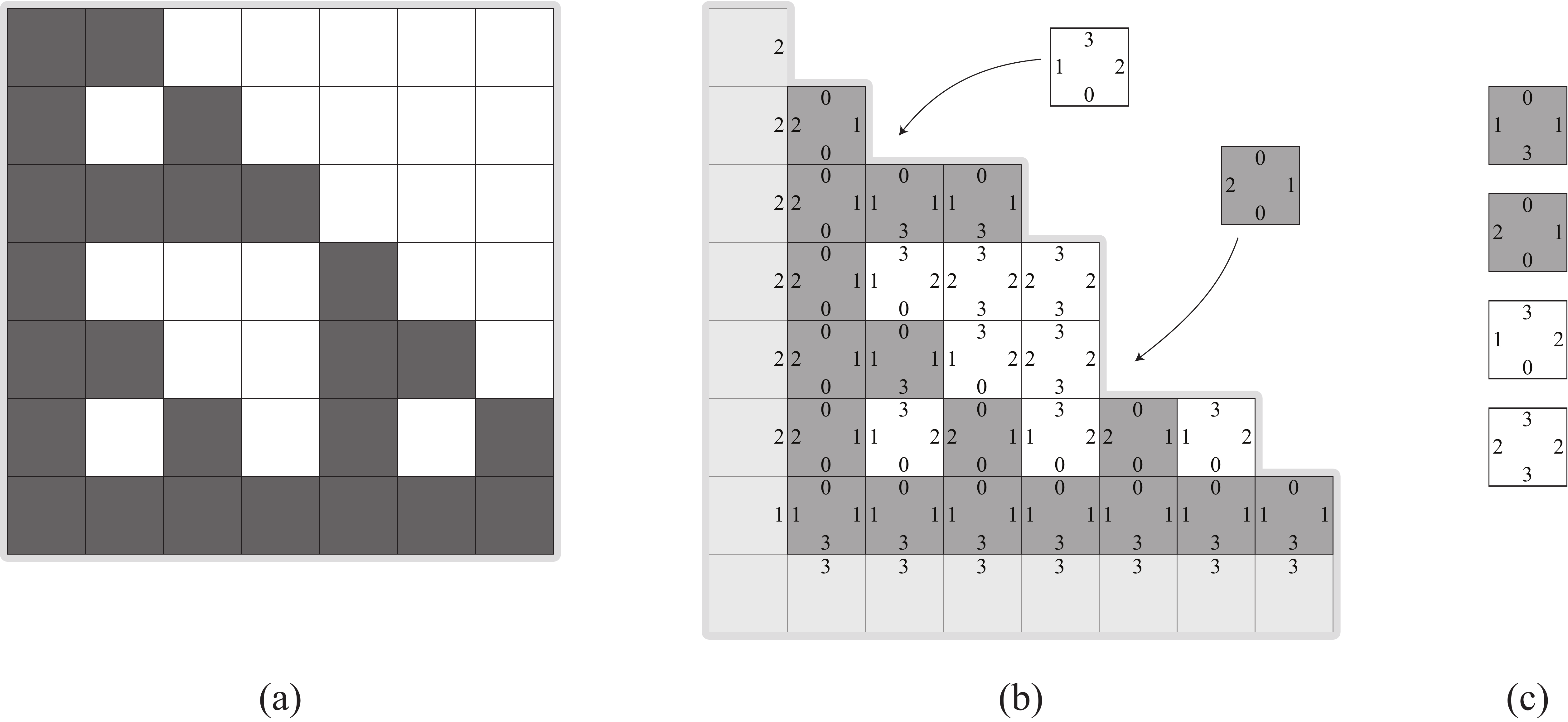}
  \caption{(a) A finite subset of the discrete Sierpinski
triangle pattern. This 2-colouring of the set $[7]\times[7]$ defines an
instance of the PATS problem. (b) Assembling the Sierpinski triangle pattern (see e.g. Winfree~\cite{Winf98b})
with a TAS that has an appropriate seed assembly and a (coloured) tile
set shown in (c).}
  \label{fig:sierpinski}
\end{figure}

In the literature, the seed assembly of a TAS is often taken to be a
single seed tile~\cite{RoWi00} whereas we consider an L-shaped seed
assembly. The boundaries can always be self-assembled using $m+n+1$
different tiles with strength-2 glues, but we wish to make a clear
distinction between the complexity of constructing the boundaries and
the complexity of the 2D pattern itself. Moreover, in some
experimental designs for DNA tile assembly systems, such as that by
Fujibayashi et al.~\cite{FHPW08}, the implementation of seed
structures by the DNA origami technique~\cite{Roth06} allows the
creation of such complete boundary conditions in a natural way.

Due to constraint~P1 the self-assembly process proceeds in a
uniform manner directed from south-west to north-east. This paves the
way for a simple characterisation of deterministic TASs in the context
of the PATS problem.
\begin{proposition} \label{prop:determinism} Solutions $\mathscr{T} =
(T,\mathcal{S},s,2)$ of the PATS problem are deterministic precisely
when for each pair of glue types $(\sigma_1,\sigma_2)\in\Sigma^2$
there is at most one tile type $t\in T$ such that $\sigma_S(t)=\sigma_1$
and $\sigma_W(t)=\sigma_2$.
\end{proposition}

A simple observation reduces the work needed in finding minimal
solutions of the PATS problem.
\begin{lemma} The minimal solutions of the PATS problem are
deterministic TASs.
\end{lemma}
\begin{proof} For the sake of contradiction, suppose that $\mathscr{N}
= (T,\mathcal{S},s,2)$ is a minimal solution to a PATS problem
instance and that $\mathscr{N}$ is not deterministic. By the above
proposition, let tiles $t_1,t_2\in T$ be such that
$\sigma_S(t_1)=\sigma_S(t_2)$ and
$\sigma_W(t_1)=\sigma_W(t_2)$. Consider the simplified TAS
$\mathscr{N}' = (T\smallsetminus\{t_2\},\mathcal{S},s,2)$. We show
that this, too, is a solution to the PATS problem, which violates the
minimality of $|T|$.

Suppose $\mathcal{A}\in\Term\mathscr{N}'$. If
$\mathcal{A}\notin\Term\mathscr{N}$, then some $t\in T$ can be used to
extend $\mathcal{A}$ in $\mathscr{N}$. If $t\in
T\smallsetminus\{t_2\}$, then $t$ could be used to extend
$\mathcal{A}$ in $\mathscr{N}'$, so we must have $t=t_2$. But since
new tiles are always attached by binding to south and west sides of
the tile, $\mathcal{A}$ could then be extended by $t_1$ in
$\mathscr{N}'$. Thus, we conclude that
$\mathcal{A}\in\Term\mathscr{N}$ and furthermore
$\Term\mathscr{N}'\subseteq\Term\mathscr{N}$. This demonstrates that
$\mathscr{N}'$ has property P2. The properties P1
and P3 can be readily seen to hold for $\mathscr{N}'$ as
well. In terms of $|T|$ we have found a more optimal solution---and a
contradiction.
\end{proof}
We consider only deterministic TASs in the sequel.

%% file: search_space.tex
Let $X$ be the family of partitions of the set $[m]\times[n]$.  Partition
$P$ is \emph{coarser} than partition~$P'$ (or $P'$ is a
\emph{refinement} of $P$), denoted $P \sqsubseteq P'$, if
\begin{equation*}
  \forall p'\in P'\colon\enspace \exists p\in P\colon\enspace p'\subseteq p.
\end{equation*}
Now, $(X,\sqsubseteq)$ is a partially ordered set, and in fact, a
lattice. Note that $P\sqsubseteq P'$ implies $|P|\leq |P'|$.

A colouring $c \colon [m]\times[n] \to [k]$ induces a partition
$P(c)=\{c^{-1}(i)\enspace|\enspace i\in [k]\}$ of the set
$[m]\times[n]$.  In addition, since every (deterministic) solution
$\mathscr{T} = (T,\mathcal{S},s,2)$ of the PATS problem uniquely
produces some assembly $\mathcal{A}$, we associate with $\mathscr{T}$
a partition $P(\mathscr{T})$ of $[m]\times[n]$, $P(\mathscr{T}) =
\{\mathcal{A}^{-1}(t)\enspace|\enspace
t\in\mathcal{A}([m]\times[n])\}$. Here, $|P(\mathscr{T})|=|T|$ in case
all tiles in $T$ are used in the terminal assembly. Now condition
P3 in the definition of PATS
is equivalent to requiring that a TAS $\mathscr{T}$ satisfies
\begin{equation*}
  P(c) \sqsubseteq P(\mathscr{T}).
\end{equation*}

A partition $P\in X$ is \emph{constructible} if $P=P(\mathscr{T})$ for
some deterministic TAS $\mathscr{T}$ satisfying properties P1
and P2. Hence the PATS problem can be rephrased using the
family of partitions as the fundamental search space.
\begin{proposition} A minimal solution to the PATS problem corresponds
to a partition $P\in X$ such that $P$ is constructible, $P(c)
\sqsubseteq P$ and $|P|$ is minimal.
\end{proposition}

\begin{figure}[t]
\centering
  \includegraphics[width=9cm]{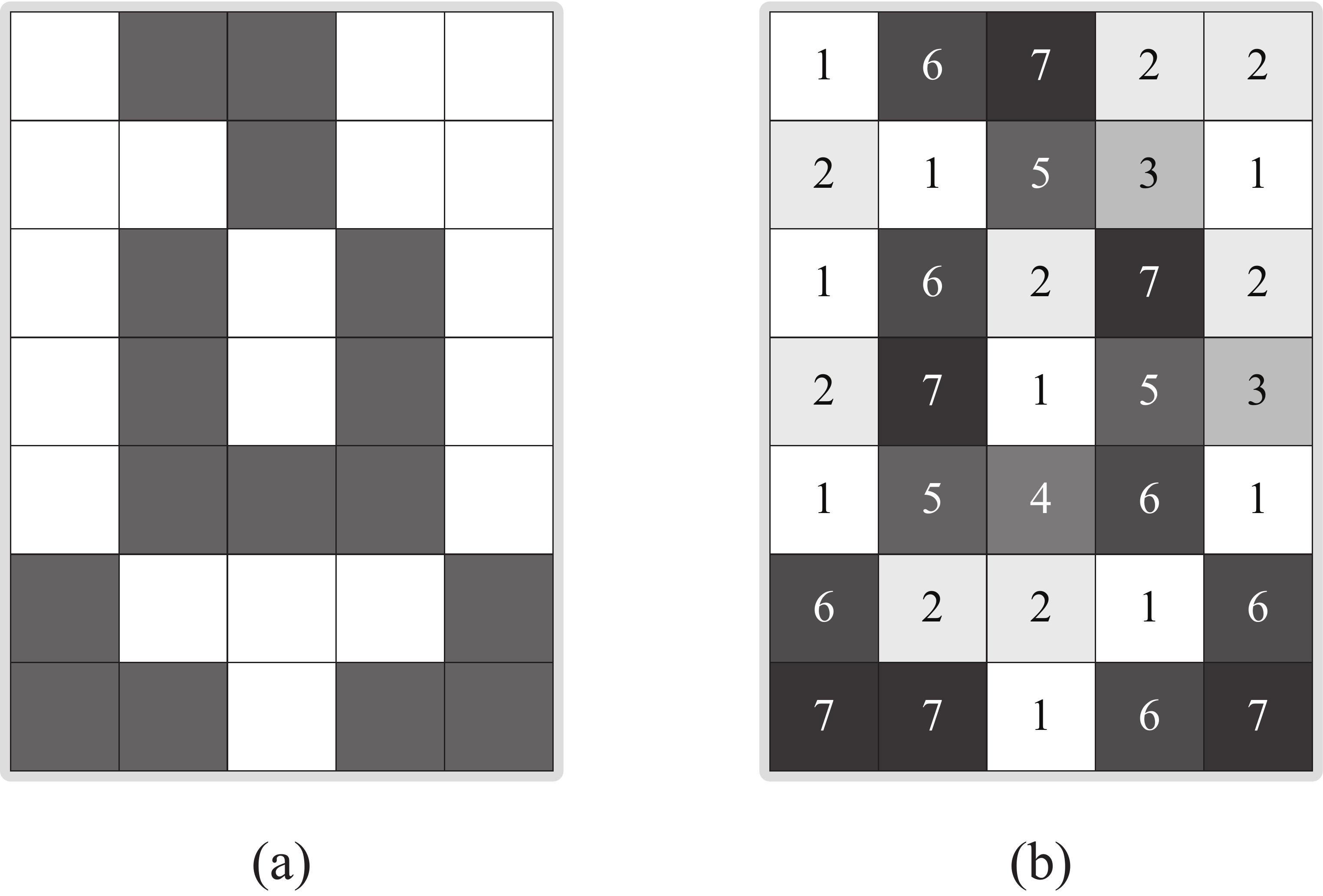}
  \caption{(a) Partition $A$. (b) A partition $M$ that is a refinement
of $A$ with $|M|=7$ partition classes.}
  \label{fig:pattern}

\end{figure}

For example, the 2-coloured pattern in Figure~\ref{fig:pattern}(a)
defines a 2-class partition~$A$. The 7-class partition~$M$ in
Figure~\ref{fig:pattern}(b) is a refinement of $A$ ($A\sqsubseteq M$)
and in fact, $M$ is constructible (see Figure~\ref{fig:construction}(b))
and corresponds to a minimal solution of the PATS problem instance defined by
the pattern~$A$.

\subsection{Determining Constructibility}\label{sec:construct}

In this section, we give an algorithm for deciding the constructibility
of a given partition in polynomial time. To do this, we use the
concept of \emph{most general} (or least constraining) \emph{tile
  assignments}. For simplicity, we assume the set of glue labels
$\Sigma$ to be infinite.

\begin{definition}
  Given a partition $P$ of the set $[m]\times[n]$, a \emph{most
    general tile assignment} (MGTA) is a function $f\colon P\to\Sigma^4$
  such that
\begin{enumerate}[{A}1. ]
\item When every position in $[m]\times[n]$ is assigned a tile type
  according to $f$, any two adjacent positions agree on the glue type
  of the side between them.
\item For all assignments $g\colon P\to\Sigma^4$ satisfying A1 we
  have\footnote{To shorten the notation, we write $f(p)_D$ instead of
    $\sigma_D(f(p))$.}
	\begin{equation*}
          f(p_1)_{D_1}=f(p_2)_{D_2} \quad \implies \quad g(p_1)_{D_1}=g(p_2)_{D_2}
	\end{equation*}
	for all $(p_1,D_1),(p_2,D_2)\in P\times\mathcal{D}$.
\end{enumerate}
\end{definition}

To demonstrate this concept, we present a most general tile assignment
$f\colon I\to\Sigma^4$ for the \emph{initial partition}
$I=\{\{a\}\enspace|\enspace a\in[m]\times[n]\}$ in Figure
\ref{fig:construction}(a) and an MGTA for the partition of Figure
\ref{fig:pattern}(b) in Figure \ref{fig:construction}(b).

\begin{figure}[t]
  \centering
  \includegraphics[width=14cm]{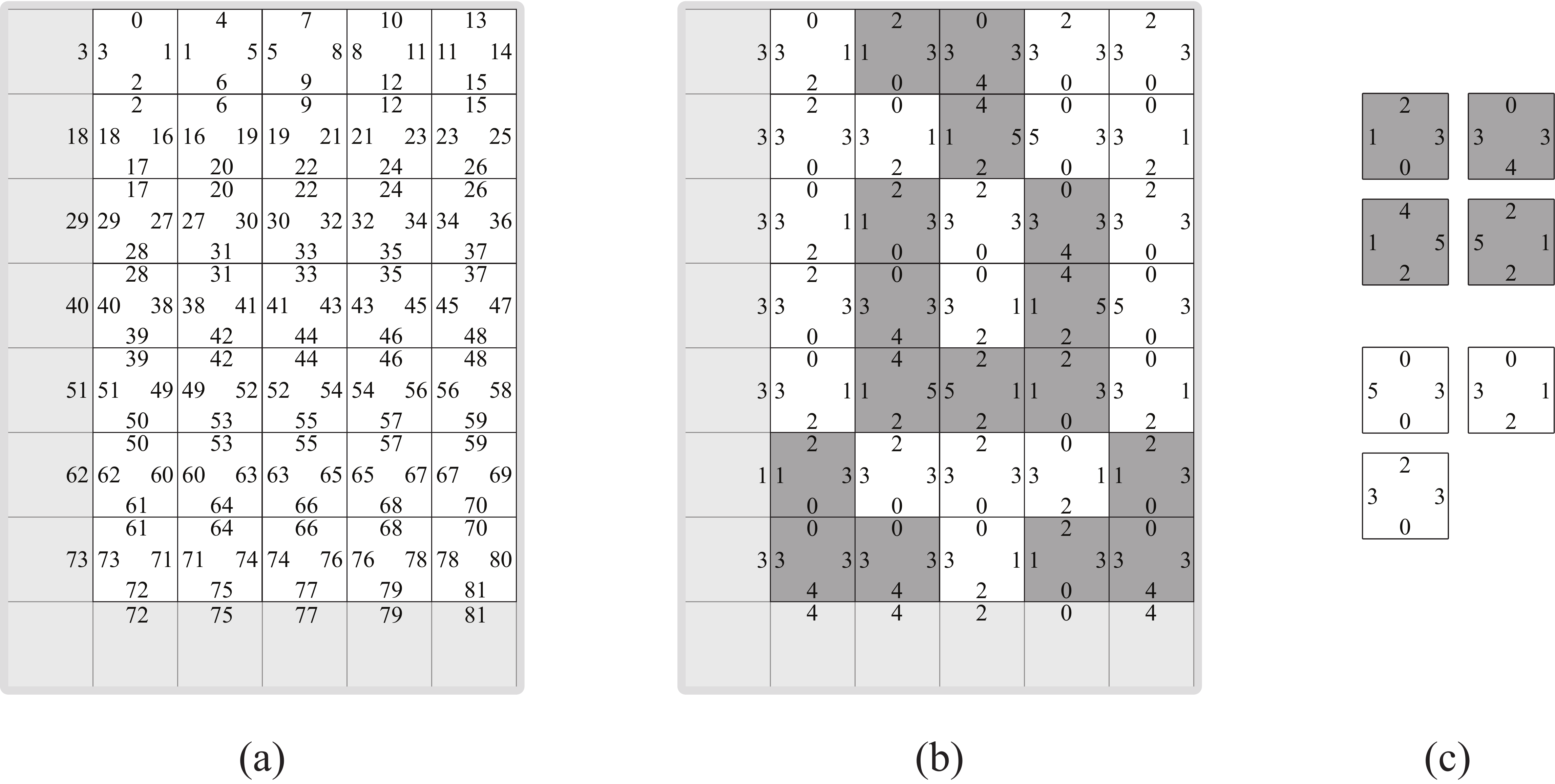}
  \caption{(a) A MGTA for the constructible initial partition $I$
    (with a seed assembly in place). (b) Finished assembly for the
    pattern from Figure \ref{fig:pattern}(a). The tile set to construct
    this assembly is given in (c).}
  \label{fig:construction}
\end{figure}

Given a partition $P\in X$ and a function $f\colon P\to\Sigma^4$, we say
that $g\colon P\to\Sigma^4$ is obtained from $f$ by \emph{merging glues $a$
  and $b$} if for all $(p,D)\in P\times\mathcal{D}$ we have
\begin{equation*}
  g(p)_D=\left\{
  \begin{aligned}
    a,\qquad&\text{if }f(p)_D=b\\
    f(p)_D,\qquad&\text{otherwise}
  \end{aligned}\right..
\end{equation*}

A most general tile assignment for a partition $P\in X$ can be found
as follows. We start with a function $f_0\colon P\to\Sigma^4$ that assigns
to each tile edge a unique glue type, or in other words, a function
$f_0$ such that the mapping $(p,D)\mapsto f_0(p)_D$ is injective. Next,
we go through all pairs of adjacent positions in $[m]\times[n]$ in
some order and require their matching sides to have the same glue
type by merging the corresponding glues. This process generates a
sequence of functions $f_0,f_1,f_2,\ldots,f_N=f$ and terminates after
$N\leq 2mn$ steps.

\begin{lemma} \label{le:mgta}
The above algorithm generates a most general tile assignment.
\end{lemma}
\begin{proof}
  By the end, we are left with a function $f$ that satisfies property
  A1 by construction. To see why property A2 is
  satisfied, we again use the language of partitions.

  Any tile assignment on $P$ gives rise to a set of equivalence
  classes (or a partition) on $P\times\mathcal{D}$: class-direction
  pairs that are assigned the same glue type reside in the same
  equivalence class. The initial assignment $f_0$ gives each
  class-direction pair a unique glue type, and thus corresponds to
  the initial partition $J=\{\{a\}\enspace|\enspace a\in
  P\times\mathcal{D}\}$. In the algorithm, any glue merging operation
  corresponds to the combining of two equivalence classes.

  The algorithm goes through a list of pairs
  $\{\{a_i,b_i\}\}_{i=0}^{N-1}$ of elements from $P\times\mathcal{D}$
  that are required to have the same glue type. In this way, the list
  records necessary conditions for property A1 to hold. This
  is to say that every tile assignment satisfying A1 has to
  correspond to a partition of $P\times\mathcal{D}$ that is coarser
  than each of the partitions in
  $\mathcal{L}=\{J[a_i,b_i]\}_{i=0}^{N-1}$, where $J[a,b]$ is the
  partition obtained from the initial partition by combining classes
  $a$ and $b$. Since the set $(P\times\mathcal{D}, \sqsubseteq)$ is a
  lattice, there exists a unique greatest lower bound $\inf
  \mathcal{L}$ of the partitions in $\mathcal{L}$. This is exactly the
  partition that the algorithm calculates in the form of the
  assignment $f$. As a greatest lower bound, $\inf \mathcal{L}$ is
  finer
  than any partition corresponding to an assignment satisfying
  A1, but this is precisely the requirement for condition A2.
\end{proof}

The above analysis also gives the following.
\begin{corollary}
  For a given partition of $[m]\times[n]$, MGTAs are unique up to
  relabelling of the glue types.
\end{corollary}
Thus, for each partition $P\in X$, we take \emph{the MGTA for $P$} to be some canonical representative from the class of MGTAs for $P$.

For efficiency purposes, it is worth mentioning that MGTAs can be
generated iteratively: A partition $P\in X$ can be obtained by
repeatedly combining partition classes starting from the initial partition $I$:
\begin{equation*}
  I = P_1 \sqsupseteq P_2 \sqsupseteq \cdots \sqsupseteq P_N = P.
\end{equation*}
As a base case, an MGTA for $I$ can be computed by the above
algorithm. An MGTA for each $P_{i+1}$ can be computed from an MGTA for
the previous partition $P_i$ by just a small modification: Let an MGTA
$f_i\colon P_i\to\Sigma^4$ be given for $P_i$ and suppose $P_{i+1}$ is
obtained from $P_i$ by combining classes $p_1,p_2\in P_i$. Now, an MGTA
$f_{i+1}$ for $P_{i+1}$ can be obtained from $f_i$ by \emph{merging
  tiles} $f_i(p_1)$ and $f_i(p_2)$, that is, merging the glue types on
the four corresponding sides.

We now give the conditions for a partition to be constructible in
terms of MGTAs.
\begin{lemma}\label{le:construct}
  A partition $P\in X$ is constructible iff the MGTA $f\colon P\to\Sigma^4$
  for $P$ is injective and the tile set $f(P)$ is deterministic in the
  sense of Proposition~\ref{prop:determinism}.
\end{lemma}
\begin{proof}
  ``$\Rightarrow$'': Let $P\in X$ be constructible and let the MGTA
  $f\colon P\to\Sigma^4$ for $P$ be given. Let $\mathscr{T}$ be a
  deterministic TAS such that $P(\mathscr{T})=P$. The uniquely
  produced assembly of $\mathscr{T}$ induces a tile assignment
  $g\colon P\to\Sigma^4$ that satisfies property A1. Now using
  property~A2 for the MGTA $f$ we see that any violation of
  the injectivity of $f$ or any violation of the determinism of the
  tile set $f(P)$ would imply such violations for $g$. But since $g$
  corresponds to a constructible partition, no violations can occur
  for $g$ and thus none for~$f$.

  ``$\Leftarrow$'': Let $f\colon P\to\Sigma^4$ be an injective MGTA with
  deterministic tile set $f(P)$. Because $f(P)$ is deterministic,
  we can choose glue types for a seed assembly $\mathcal{S}$ so that the
  westernmost and southernmost tiles fall into place according to $f$
  in the self-assembly process. The TAS $\mathscr{T}=(f(P),
  \mathcal{S},s,2)$, with appropriate glue strengths $s$, then
  uniquely produces a terminal assembly that agrees with $f$ on
  $[m]\times[n]$. This gives $P(\mathscr{T}) \sqsubseteq P$, but since
  $f$ is injective, $|P|=|f(P)|=|P(\mathscr{T})|$ and so
  $P(\mathscr{T})=P$.
\end{proof}

\begin{figure}[t]
  \centering
  \includegraphics[width=\textwidth]{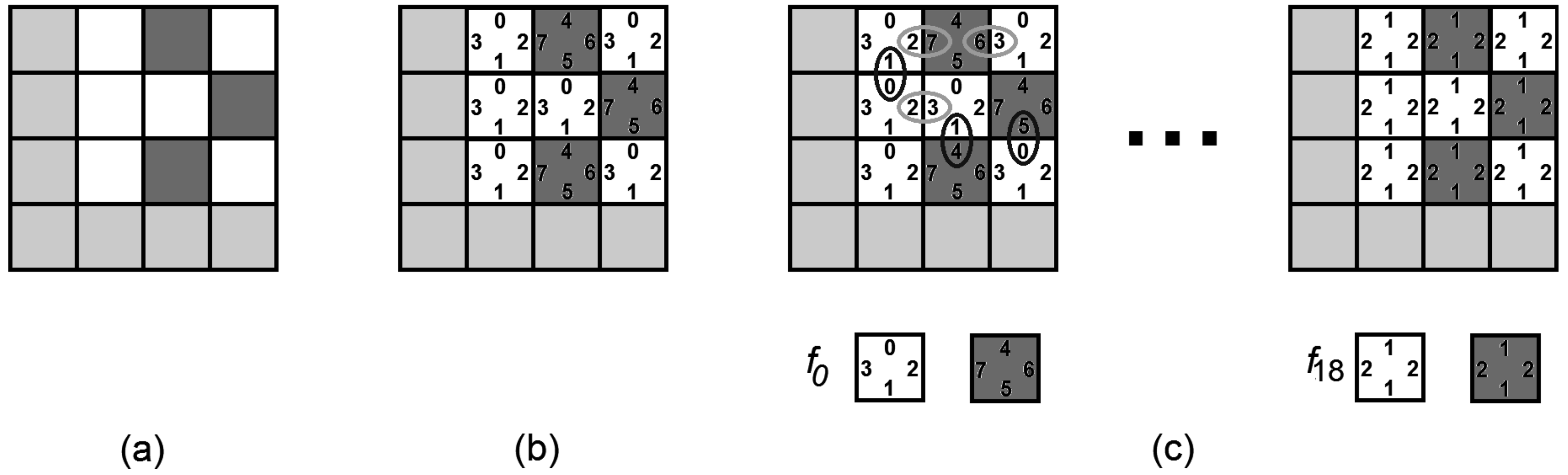}
  \caption{(a) A 2-coloured pattern. (b) The 2-class partition generated by the two colours and the initialisation of the procedure for finding an MGTA for this partition. (c) As a result of the MGTA generation procedure we obtain two tiles which have the same glues on all corresponding edges.}
  \label{fig:part}
\end{figure}

In order to understand the result of Lemma~\ref{le:construct} better,
let us consider the 2-coloured pattern in Figure~\ref{fig:part}(a),
and associate to it the 2-class partition generated by the colours of
the pattern. We can use the result of the previous lemma to show that
this partition in not constructible. Indeed, if we consider the
procedure for generating an MGTA for this partition, e.g.\
Figure~\ref{fig:part}(b) and (c), we obtain that the two tiles of the
MGTA (one coloured white, the other black) must have the same glues on
all corresponding positions. Hence the MGTA is not injective, nor
deterministic in the sense of Proposition~\ref{prop:determinism}.

%% file: complete_search.tex
We now extend the techniques of Ma and Lombardi~\cite{MaLo08} to
obtain an exhaustive branch-and-bound search method to find minimal
solutions to the PATS problem. We call this approach the
\emph{partition-search branch-and-bound} (PS-BB) algorithm. The idea
of Ma and Lombardi \cite{MaLo08} (following experimental work of Park
et al.~\cite{PYRL04}) is to start with an \emph{initial tile set} that
consists of $mn$ different tiles, one for each of the grid positions
in $[m]\times[n]$. Their algorithm then proceeds to merge tile types
in order to minimise $|T|$. We formalise this search process as an
exhaustive search in the set of all partitions of the set
$[m]\times[n]$.  In the following, we let a PATS instance be given by
a fixed $k$-coloured pattern $c\colon [m]\times[n]\to [k]$.

\begin{figure}[t]
  \centering
  \includegraphics[width=\textwidth]{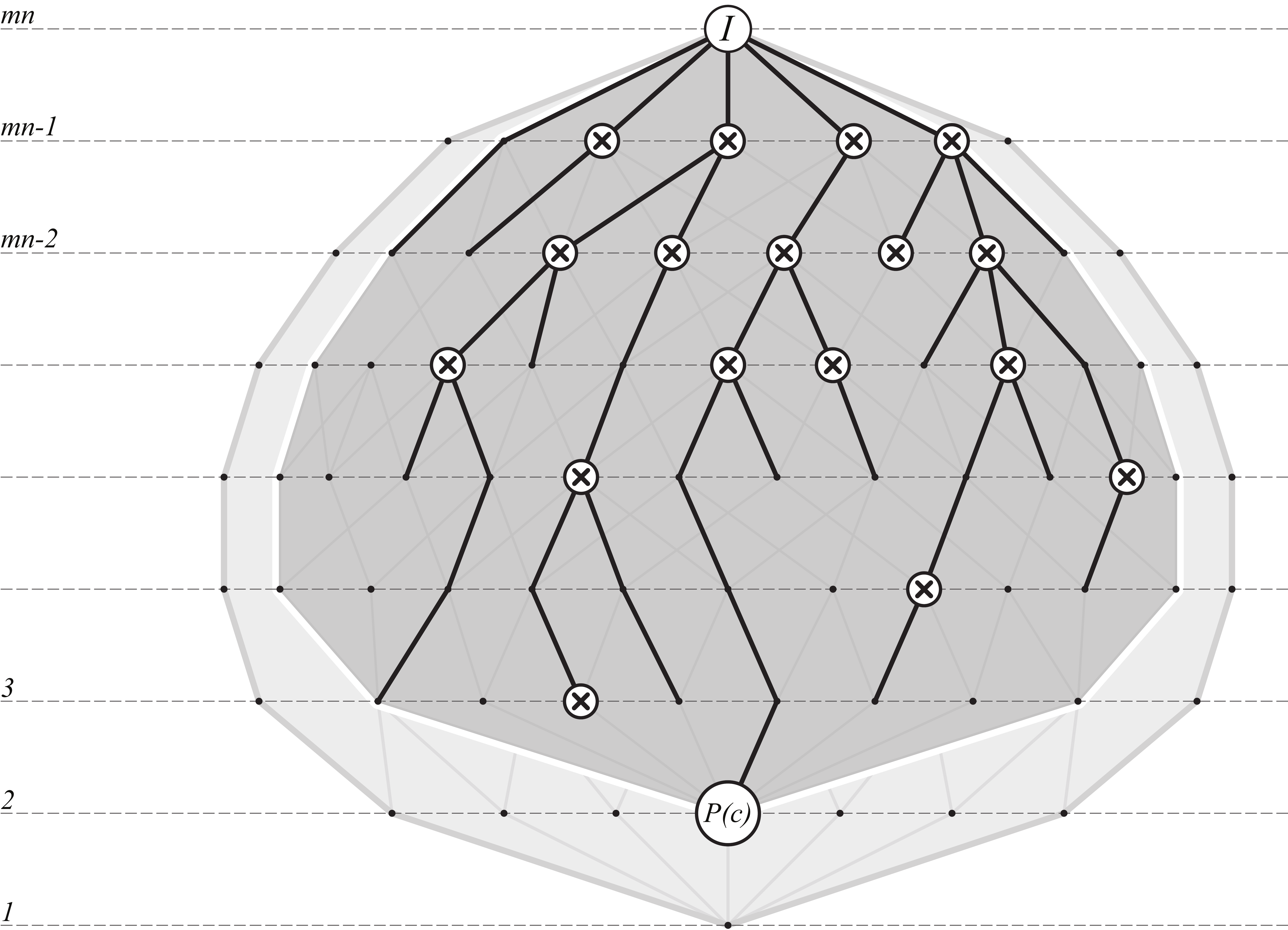}
  \caption{The search tree in the lattice $(X,\sqsubseteq)$. We start
    with the initial partition $I$ of size $|I|=mn$. The partition
    $P(c)$ defines the PATS problem instance: We search for
    constructible partitions (drawn as crosses) in the sublattice
    (shaded with darker grey) consisting of those partitions that are
    refinements of $P(c)$. The search tree branches only at the
    constructible partitions and the tree branches are
    vertex-disjoint.}
  \label{fig:lattice}
\end{figure}

The PS-BB algorithm performs an exhaustive exploration of the lattice
$(X,\sqsubseteq)$, searching for constructible partitions (see
Figure~\ref{fig:lattice}). We start with the initial partition~$I$
that is always constructible. In the search, we maintain and
incrementally update MGTAs for every partition we visit. First, we
describe simple branching rules
to obtain a rooted directed acyclic graph search structure and later give rules
to prune this DAG to a node-disjoint search tree.

The root of the DAG is taken to be the initial partition
$I$. For each partition $P\in X$ we next define the set $C(P)\subseteq
X$ of children of $P$. Our algorithm always proceeds by combining
classes of the partition currently being visited, so for each $P'\in
C(P)$ we will have $P'\sqsubseteq P$.  Say we visit a partition $P\in
X$. We have two possibilities:
\begin{enumerate}[C1. ]
\item \emph{$P$ is constructible:}
  \begin{enumerate}[1. ]
  \item If $P$ is not a refinement of the target pattern $P(c)$, that
    is if $P(c)\not\sqsubseteq P$, we can drop this branch of the
    search, since no possible descendant $P'\sqsubseteq P$ can be a
    refinement of $P(c)$ either.
  \item In case $P(c)\sqsubseteq P$, we can use the MGTA for $P$ to
    give a concrete solution to the PATS problem instance defined by
    the colouring $c$. To continue the search and to find further
    improved solutions we consider each pair of classes
    $\{p_1,p_2\}\subseteq P$ in turn and recursively visit the
    partition $P[p_1,p_2]$ where the two classes are combined. In fact,
    by the above analysis, it is sufficient to consider only pairs of
    the same colour. So, in this case,
    \begin{equation*}
      C(P)=\{P[p_1,p_2]\enspace|\enspace p_1,p_2\in P,\enspace p_1\neq p_2,\enspace \exists k\in P(c)\colon p_1,p_2\subseteq k\}.
    \end{equation*}
  \end{enumerate}
\item \emph{$P$ is not constructible:} In this case the MGTA $f$ for
  $P$ gives $f(p_1)_S=f(p_2)_S$ and $f(p_1)_W=f(p_2)_W$ for some classes
  $p_1\neq p_2$. We continue the search from partition $P[p_1,p_2]$.
\end{enumerate}

To guarantee that our algorithm finds the optimal solution in the case
C2 above, we need the following.
\begin{lemma}
  Let $P\in X$ be a non-constructible partition, $f$ the MGTA for $P$
  and $p_1,p_2\in P$, $p_1\neq p_2$, classes such that
  $f(p_1)_S=f(p_2)_S$ and $f(p_1)_W=f(p_2)_W$. For all constructible
  $C\sqsubseteq P$ we have $C \sqsubseteq P[p_1,p_2]$.
\end{lemma}
\begin{proof}
  Let $P$, $f$, $p_1$ and $p_2$ be as in the statement of the
  lemma. Let $C\sqsubseteq P$ be a constructible partition and
  $g\colon C\to\Sigma^4$ the MGTA for $C$. Since $C$ is coarser than
  $P$ we can obtain from $g$ a tile assignment $g'\colon P\to\Sigma^4$
  such that $g'(p)=g(q)$, where for every $p \in P$, $q\in C$ is the
  unique class for which $p\subseteq q$. The assignment $g'$ has
  property A1 and so using A2 for the MGTA $f$ we
  get that
  \[
    f(p_1)_S=f(p_2)_S \enspace \& \enspace f(p_1)_W=f(p_2)_W
    \enspace\implies\enspace
    g'(p_1)_S=g'(p_2)_S \enspace \& \enspace g'(p_1)_W=g'(p_2)_W.
  \]
  Now, since $C$ is constructible, the identities $g(q_1)_S=g(q_2)_S$
  and $g(q_1)_W=g(q_2)_W$ can not hold for any two different classes
  $q_1,q_2\in C$. Looking at the definition of $g'$, we conclude that
  $p_1\subseteq q$ and $p_2 \subseteq q$ for some $q \in C$. This
  demonstrates $C \sqsubseteq P[p_1,p_2]$.
\end{proof}

\subsection{Pruning the DAG to a Search Tree}

Computational resources should be saved by not visiting any partition
twice. To keep the branches in our search structure node-disjoint, we
maintain a list of graphs that store restrictions on the choices the
search can make.

For each partition $P \sqsupseteq P(c)$ we associate
a family of undirected graphs $\{G^P_k\}_{k\in P(c)}$, one for each
colour class of the pattern $P(c)$. Every class in $P$ is represented
by a vertex in the graph corresponding to the colour of the class. More
formally, the vertex set $V(G^P_k)$ of the graph $G^P_k$ is taken to
be those classes $p\in P$ for which $p\subseteq k$. (So now,
$\bigcup_{k\in P(c)} V(G^P_k) = P$.)
An edge $\{p_1,p_2\}\in E(G^P_k)$ indicates that the classes $p_1$ and~$p_2$ are not
allowed ever to be combined in the search branch in question. When we
start our search with the initial partition $I$, the edge sets are
initially empty, $E(G^I_k)=\varnothing$. At each partition $P$, the
graphs $\{G^P_k\}_{k\in P(c)}$ have been determined inductively and
the graphs for those children $P'\in C(P)$ that we visit are defined
as follows.

\begin{enumerate}[D1. ]
\item \emph{If $P$ is constructible:} We choose some ordering
$\{p_i,q_i\}$, $i=1,\ldots,N$ of similarly coloured pairs of
classes. Define $l_i\in P(c)$, $1\leq i \leq N$ to be the colour of the
pair $\{p_i,q_i\}$, so that $p_i,q_i\subseteq l_i$. Now, we visit a
partition $P[p_i,q_i]$ if and only if $\{p_i,q_i\}\notin E(G^P_{l_i})$. If we
decide to visit a child partition $P'=P[p_j,q_j]$, we define the edge
sets $\{E(G^{P'}_k)\}_{k\in P(c)}$ as follows:
  \begin{enumerate}[1. ]
  \item We start with the graphs $\{G^P_k\}_{k\in P(c)}$ and add the
    edges $\{p_i,q_i\}$ for all $1\leq i < j$ to their corresponding
    graphs. Call the resulting graphs $\{G^\star_k\}_{k\in P(c)}$.
  \item Finally, as we combine the classes $p_j$ and $q_j$ to obtain the
    partition $P[p_j,q_j]$, we merge the vertices $p_j$ and $q_j$ in
    the graph $G^\star_{l_j}$ (after merging, the neighbourhood of the
    new vertex $p_j \cup q_j$ is the union of the neighbourhoods for
    $p_j$ and $q_j$ in $G^\star_{l_j}$). The graphs
    $\{G^{P'}_k\}_{k\in P(c)}$ follow as a result.
  \end{enumerate}
\item \emph{If $P$ is not constructible:} Here, the MGTA for $P$
  suggests a single child partition $P'=P[p_1,p_2]$ for some
  $p_1,p_2\subseteq l\in P(c)$. If $\{p_1,p_2\}\in E(G^P_l)$, we
  terminate this branch of the search. Otherwise, we define the graphs
  $\{G^{P'}_k\}_{k\in P(c)}$ to be the graphs $\{G^P_k\}_{k\in P(c)}$,
  except that in $G^{P'}_l$ the vertices $p_1$ and $p_2$ are merged.
\end{enumerate}

One can see that the outcome of this pruning process is a search tree
that has node-disjoint branches and one in which every possible
constructible partition is still guaranteed to be found.
Figure~\ref{fig:lattice} presents a sketch of the search tree.

Note that we
are not usually interested in finding every constructible partition
$P\in X$, but only in finding a minimal one (in terms of $|P|$). Next,
we give an efficient method to lower-bound the partition sizes of a
given search branch.

\subsection{The Bounding Function}

Given a root $P\in X$ of some subtree of the search tree, we ask: What
is the smallest partition that can be found from this subtree?  The
nodes in the subtree rooted at $P$ comprise those partitions
$P'\sqsubseteq P$ that can be obtained from $P$ by merging pairs of
classes that are not forbidden by the graphs $\{G^P_k\}_{k\in
  P(c)}$. This merging process halts precisely when all the graphs
$\{G^{P'}_k\}_{k\in P(c)}$ have been reduced into cliques. As is well known
(and easy to see),
the size of the smallest clique that a graph $G$ can be turned
into by merging non-adjacent vertices is given by the \emph{chromatic
  number}\footnote{The chromatic number of a graph $G$ is the smallest
  number of colours $\chi(G)$ needed to colour the vertices of $G$ so
  that no two adjacent vertices share the same colour.} $\chi(G)$ of
the graph $G$. This immediately gives the following.
\begin{proposition}
  For every $P'\sqsubseteq P$ in the subtree rooted at $P$ and
  constrained by $\{G^P_k\}_{k\in P(c)}$, we have
\begin{equation*}
  \sum_{k\in P(c)}\chi(G^P_k) \enspace \leq \enspace |P'|.
\end{equation*}
\end{proposition}

Determining the chromatic number of an arbitrary graph is an
\textsf{NP}-hard problem. Fortunately, we can restrict our graphs to
be of a special form: graphs that consist only of a clique and some
isolated vertices. For these graphs, the chromatic numbers are given
by the sizes of the cliques.

To see how to maintain graphs in this form, consider as a base case
the initial partition~$I$. Here, $E(G^I_k)=\varnothing$ for all $k\in
P(c)$, so $G^I_k$ is of our special form---it has a clique of size~1.
For a general partition~$P$, we go through the branching rules
D1--D2.

\begin{enumerate}[D1: ]
\item \emph{$P$ is constructible:} Since we are allowed to choose an
  arbitrary ordering $\{p_i,q_i\}$, $i=1,\ldots,N$, for the children
  $P[p_i,q_i]$, we design an ordering that preserves the special form
  of the graphs. For a graph $G$ of our special form, let
  $K(G)\subseteq V(G)$ consist of those vertices that are part of the
  clique in $G$. In the algorithm, we first set $H_k = G^P_k$ for all
  $k\in P(c)$ and repeat the following process until every graph $H_k$
  is a complete clique.
  \begin{enumerate}[1. ]
  \item Pick some colour $k\in P(c)$ and an isolated vertex $v \in
    V(H_k)\smallsetminus K(H_k)$.
  \item Process the pairs $\{v,u\}$ for all $u\in K(H_k)$ in some
    order. By the end, update $H_k$ to include all the edges $\{v,u\}$
    that were just processed (the size of the clique in $H_k$
    increases by one).
  \end{enumerate}
  A moment's inspection reveals that when the graphs $G^P_k$ are of
  our special form, so are all of the derived graphs passed on to the
  children of $P$.
\item \emph{$P$ is not constructible:} If the algorithm decides to
  continue the search from a partition $P'=P[p_1,p_2]$, for some
  $p_1,p_2\subseteq l\in P(c)$, we have $\{p_1,p_2\}\notin
  E(G^P_l)$. This means that either $p_1,p_2\in V(G^P_l)\smallsetminus
  K(G^P_l)$, in which case we are merging two isolated vertices, or
  one of $p_1$ and $p_2$ is part of the clique $K(G^P_l)$, in which
  case we merge an isolated vertex to the clique. In both cases, we
  maintain the special form in the graphs $\{G^{P'}_k\}_{k\in P(c)}$.
\end{enumerate}

\subsection{Traversing the Search Tree}

When running a branch-and-bound algorithm we maintain a ``current best
solution'' discovered so far as a global variable. This solution gives
an upper bound for the minimal value of the tile set size and can be
used to prune such search branches that are guaranteed (by the
bounding function) to only yield solutions worse than the current
best.   There are two general strategies to traverse a
branch-and-bound search tree: \emph{Depth-First Search} and
\emph{Best-First Search}~\cite{ClPe99}. Our description of the search
tree for the lattice $X$ is general enough to allow either of these
strategies to be used in the actual implementation of the algorithm.
In the following section we give performance data on our DFS
implementation of the PS-BB algorithm.

\subsection{Results}
\label{sec:results}

\begin{figure}[t]
  \centering
  \includegraphics[width=\textwidth]{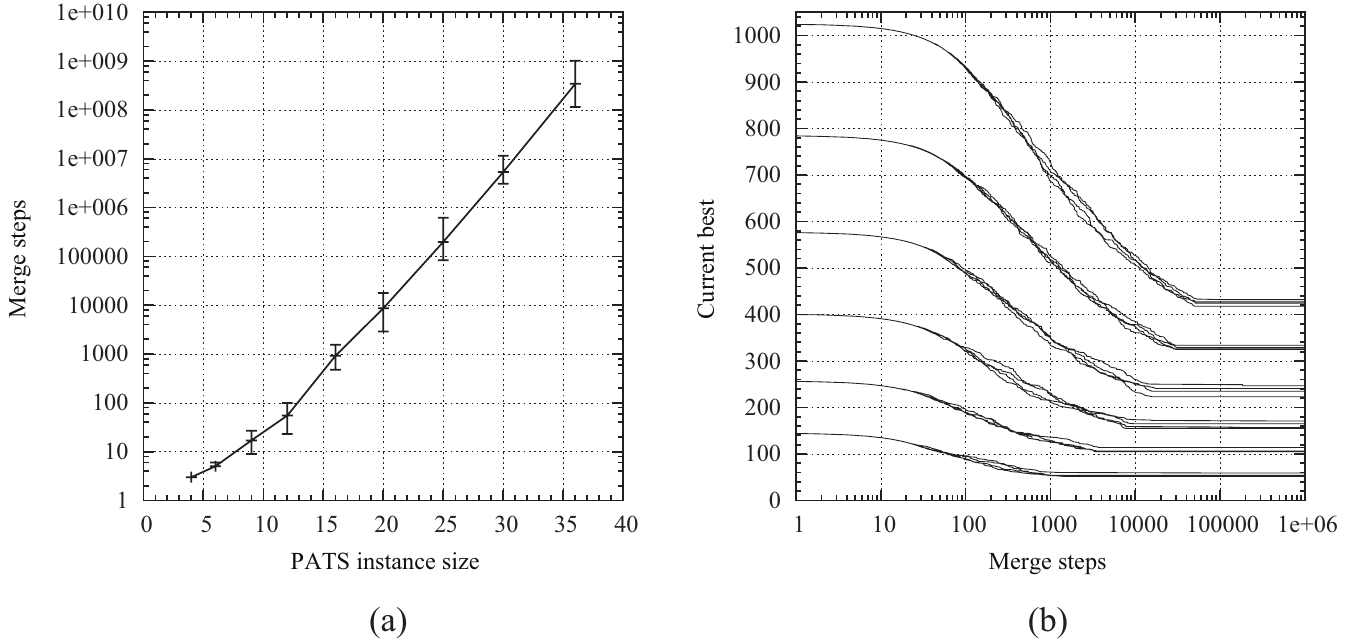}
  \caption{(a)~Running time of the PS-BB algorithm (as measured by the
    number of merge operations) to solve random 2-coloured
    near-square-shaped instances of the PATS problem. (b)~Evolution of
    the tile set size of the ``current best solution'' for the PATS problem for
    random 2-coloured instances of sizes from $12\times12$ up to $32\times32$.}
  \label{fig:runtime}
\end{figure}
\begin{figure}[t]
  \centering
  \includegraphics[width=\textwidth]{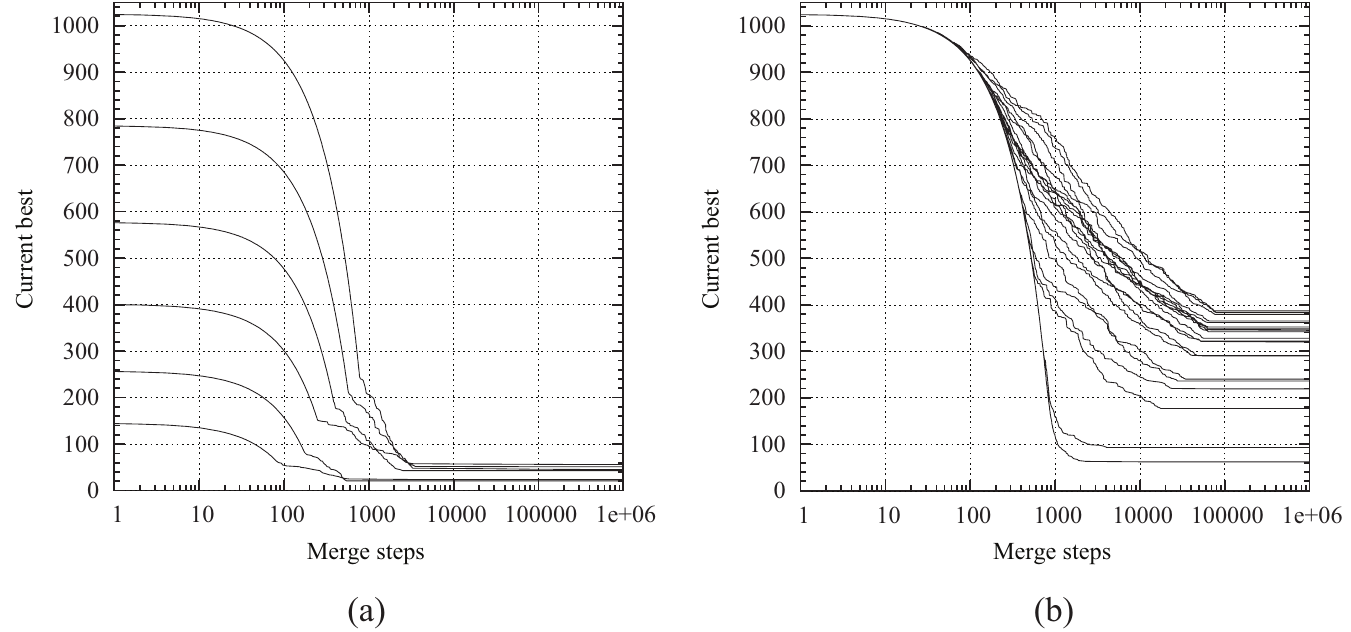}
  \caption{Evolution of the ``current best solution'' of the PS-BB
    algorithm for (a) the Sierpinski triangle pattern and for (b) the binary
    counter pattern. The lines represent (a) single runs for instance sizes
    from $12\times12$ up to $32\times32$ and (b) several runs for instance size
    $32\times32$. Randomisation in the DFS has a clear effect on
    the performance of the algorithm in the case of the binary counter
    pattern, but not in the case of the Sierpinski triangle pattern.}
  \label{fig:curbest}
\end{figure}

The running time of the PS-BB algorithm is proportional---up to a
polynomial factor---to the number of partitions the algorithm
visits. Hence, we measure the running time in terms of the number of
merge operations performed in the search.  Figure~\ref{fig:runtime}(a)
presents the number of such merge operations in order to find a minimal solution
for random 2-coloured instances of the PATS problem. The algorithm was
executed for instance sizes $2\times2, 2\times 3, 3\times 3, \ldots,
5\times6$ and $6\times6$; the 20th and 80th percentiles are shown
alongside the median of 21 separate runs for each instance size. For
the limiting case $6\times6$, the algorithm spent on the order of two
hours of (median) computing time on a 2.61 GHz AMD processor.

Even though branch-and-bound search is an exact method, it can be used
to find approximate solutions by running it for a suitable length of
time. Figure~\ref{fig:runtime}(b) illustrates how the best solution
found up to a point develops as increasingly many steps of the
algorithm are run. The figure provides data on random 2-coloured
instances of sizes
$12\times12, 16\times16, 20\times20, \ldots, 32\times32$.
Because we begin our search from the initial partition,
the best solution at the first step is precisely equal to the instance
size. For each size, several different patterns were used. The
algorithm was cut off after $10^6$ steps. By this time, an approximate
reduction of 58\% in the size of the tile set was achieved (cf.\ a
reduction of 43.5\% in Ma and Lombardi~\cite{MaLo08}).

Next, we consider two well known examples of structured patterns: the
discrete Sierpinski triangle and the binary counter (see Figures
\ref{fig:pat-srp32} and~\ref{fig:pat-bc32} for $32\times32$ instances
of both patterns). A tile set of size 4 is optimal for both of these
patterns, see e.g.\ Winfree~\cite{Winf98b} or Rothemund and Winfree~\cite{RoWi00}.
First, for the Sierpinski triangle pattern, we get a tile set reduction
of well over 90\% (cf.\ 45\% in Ma and Lombardi~\cite{MaLo08}) in
Figure~\ref{fig:curbest}(a). We used the same cutoff threshold and instance
sizes as in Figure~\ref{fig:runtime}(b).

Our description of the PS-BB algorithm leaves some room for
randomisation in deciding which search branch the DFS is to
explore next. This randomisation does not seem to affect the search
dramatically if considering the Sierpinski triangle pattern---the separate
single runs in Figure~\ref{fig:curbest}(a) are representative of an
average randomised run. By contrast, for the binary counter pattern,
randomised runs for a single instance size do make a
difference. Figure~\ref{fig:curbest}(b) depicts several separate runs
for instance size $32\times32$. Here, each run brings about a
reduction in solution size that oscillates between a reduction
achieved on a random 2-coloured instance (Figure~\ref{fig:runtime}(b)) and a
reduction achieved on the Sierpinski instance
(Figure~\ref{fig:curbest}(a)). This suggests that, as is characteristic of
DFS traversal, restarting the algorithm with different random seeds
may help with large instances that have small optimal solutions.  We
explore this opportunity for efficiency improvement further in
connection to the algorithm PS-H presented in the next section.

%% file: local_search.tex
\subsection{The PS-H Algorithm Scheme}\label{sec:PSH}

The PS-BB algorithm utilises effective pruning methods to reduce the
search space. Even though it offers significant reduction in the size
of tile sets compared to earlier approaches, it is in most cases still
too slow for patterns of practical size. Often it is not important to
find a provably minimal solution, but to find a reasonably small
solution in a reasonable amount of time.  To address this objective, we
present in the following a modification of the basic PS-BB algorithm
with a number of search-guiding heuristics. We call this approach the
\emph{partition-search with heuristics} (PS-H) algorithm scheme.

Whereas the pruning methods of the PS-BB algorithm try to reduce
the size of the search space in a ``balanced'' way, the PS-H
algorithm attempts to ``greedily'' optimise the order in which the
coarsenings of a partition are explored, in the hope of being directly
led to close-to-optimal solutions. Such opportunism may be expected
to pay off in case the success probability of the greedy exploration
is sufficiently high, and the process is restarted sufficiently often,
or equivalently, several runs are explored in parallel.

The basic heuristic idea is to try to minimise the effect that a merge
operation has on partition classes other than those which are
combined. This can be achieved by preferring to merge classes already
having as many common glues as possible. In this way one hopes to
extend the number of steps the search takes before it runs into a
conflict.  For example, when merging classes $p_1$ and $p_2$ such that
$f(p_1)_N = f(p_2)_N$ and $f(p_1)_E = f(p_2)_E$, the glues on the W
and S edges of all other classes are unaffected. This way, the search
avoids proceeding to a partition which is not constructible after the
merge operation is completed. Secondarily, we prefer merging classes
which already cover a large number of sites in $[m]\times[n]$. That
is, one tries to grow a small number of large classes instead of
growing all the classes at an equal rate.

We define the concept of the \emph{number of common glues} formally as
follows.

\begin{definition}
  Given a partition $P$ and a MGTA $f$ for $P$, the \emph{number of common
  glues} between classes $p,q \in P$ is defined by the function $G\colon P \times P
  \to \{0,1,2,3,4\}$,
  \begin{equation*}
    G(p,q) = \sum_{D\in\mathcal{D}} g(f(p)_D,f(q)_D),
  \end{equation*}
  where $g(\sigma_1,\sigma_2)= 1$ if $\sigma_1=\sigma_2$
  and 0 otherwise, for all $\sigma_1, \sigma_2 \in \Sigma$.
\end{definition}

Except for the bounding function, the PS-BB algorithm allows an
arbitrary ordering $\{p_i,q_i\}, i = 1,\ldots,N$, for the children
(coarsenings) $P[p_i,q_i]$ of a constructible partition $P$.
In the PS-H algorithm, we choose the ordering using the following
heuristics. First form the set
\begin{equation*}
  H \coloneqq \{\{p,q\} \enspace|\enspace p,q\in P,\enspace p \ne q,\enspace \exists r\in P(c)\colon\enspace p,q\subseteq r\}
\end{equation*}
of class pairs of same colour, and then repeat the following process
until $H$ is empty.

\begin{enumerate}[{H}1. ]
\item Set $K \coloneqq H$.
\item Maximise the number of common glues:
  \[
  K \coloneqq \{\{p,q\}\in K \enspace|\enspace G(p,q) \ge G(u,v)\text{ for
    all }\{u,v\} \in K\}.
  \]
\item Maximise the size of the larger class:
  \[
  K \coloneqq \{\{p,q\}\in K \enspace|\enspace \max\{|p|,|q|\} \ge
  \max\{|u|,|v|\}\text{ for all }\{u,v\} \in K\}.
  \]
\item Maximise the size of the smaller class:
  \[
  K \coloneqq \{\{p,q\}\in K \enspace|\enspace \min\{|p|,|q|\} \ge
  \min\{|u|,|v|\}\text{ for all }\{u,v\} \in K\}.
  \]
\item Pick some pair $\{p,q\}\in K$ at random and visit the partition
  $P[p,q]$.
\item Remove $\{p,q\}$ from $H$:
  \[
  H \coloneqq H \smallsetminus \{\{p,q\}\}.
  \]
\end{enumerate}

The PS-H algorithm also omits the pruning process utilised by the
PS-BB algorithm. That way, it aims to get to the small solutions
quickly by reducing the computational resources used in a single merge
operation.

Since step H5 of the heuristics above leaves room for randomisation,
the PS-H algorithm performs differently with different random
seeds. While some of the randomised runs may lead to small solutions
quickly, others may get sidetracked into worthless expanses of the
solution space.  We make the best of this situation by running several
executions of the algorithm in parallel, or equivalently, restarting
the search several times with a different random seed. The notation
PS-H$_n$ denotes the heuristic partition search algorithm
with $n$ parallel search threads.  The solution found by the
PS-H$_n$ algorithm is the smallest solution found by any of
the $n$ parallel threads.

\subsection{Results}\label{sec:psh-res}

In this section, we present results on the performance of the
PS-H$_n$ algorithm for $n=1,2,4,8,16,32$
and compare it to the previous PS-BB algorithm.
We consider several different finite 2-coloured input patterns, two of
which were analysed also previously using the PS-BB algorithm:
the discrete Sierpinski triangles of sizes $32\times32$
(Figure~\ref{fig:pat-srp32}) and $64\times64$, and the binary counter of size
$32\times32$ (Figure~\ref{fig:pat-bc32}).  Furthermore, we introduce
a 2-coloured ``tree'' pattern of size $23\times23$ (Figure~\ref{fig:pat-ct23})
as well as a 15-coloured pattern of size $20\times10$ based on a CMOS full
adder design (Figure~\ref{fig:pat-fa}).\footnote{For an explanation of the
notation used in Figure~\ref{fig:pat-fa}, see Appendix~\ref{sec:cnfet}.}
While the Sierpinski triangle
and binary counter patterns are known to have minimal solutions of
4 tiles, the minimal solutions for the tree pattern and the full
adder pattern are unknown.
The experiments were conducted on a high performance computing cluster equipped
with 2.6~GHz AMD Opteron 2435 processors and Scientific Linux~6 operating system.

\begin{figure}[t!]
  \centering
  \subfigure[]{
    \includegraphics[width=0.3\textwidth]{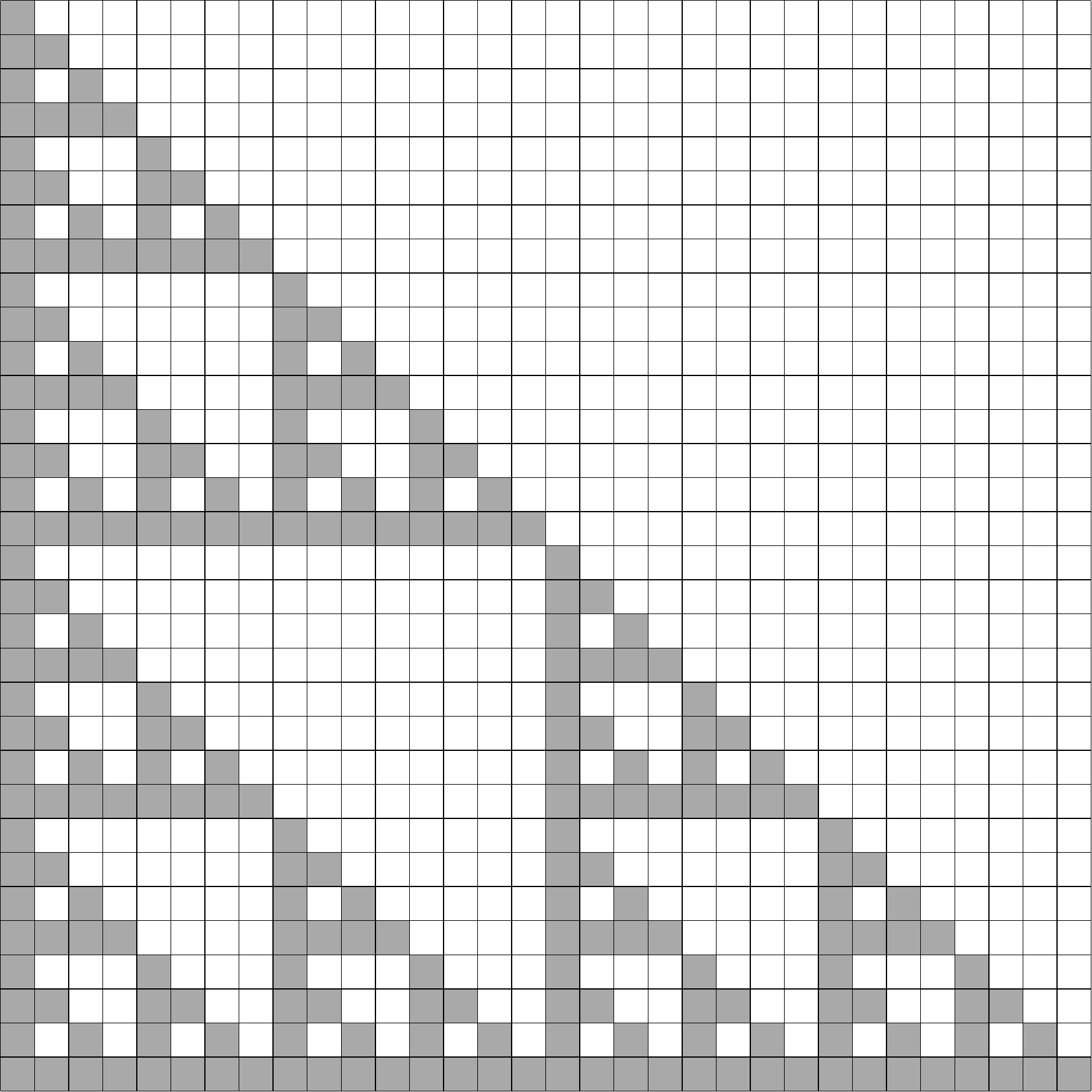}
    \label{fig:pat-srp32}
  }
  \subfigure[]{
    \includegraphics[width=0.3\textwidth]{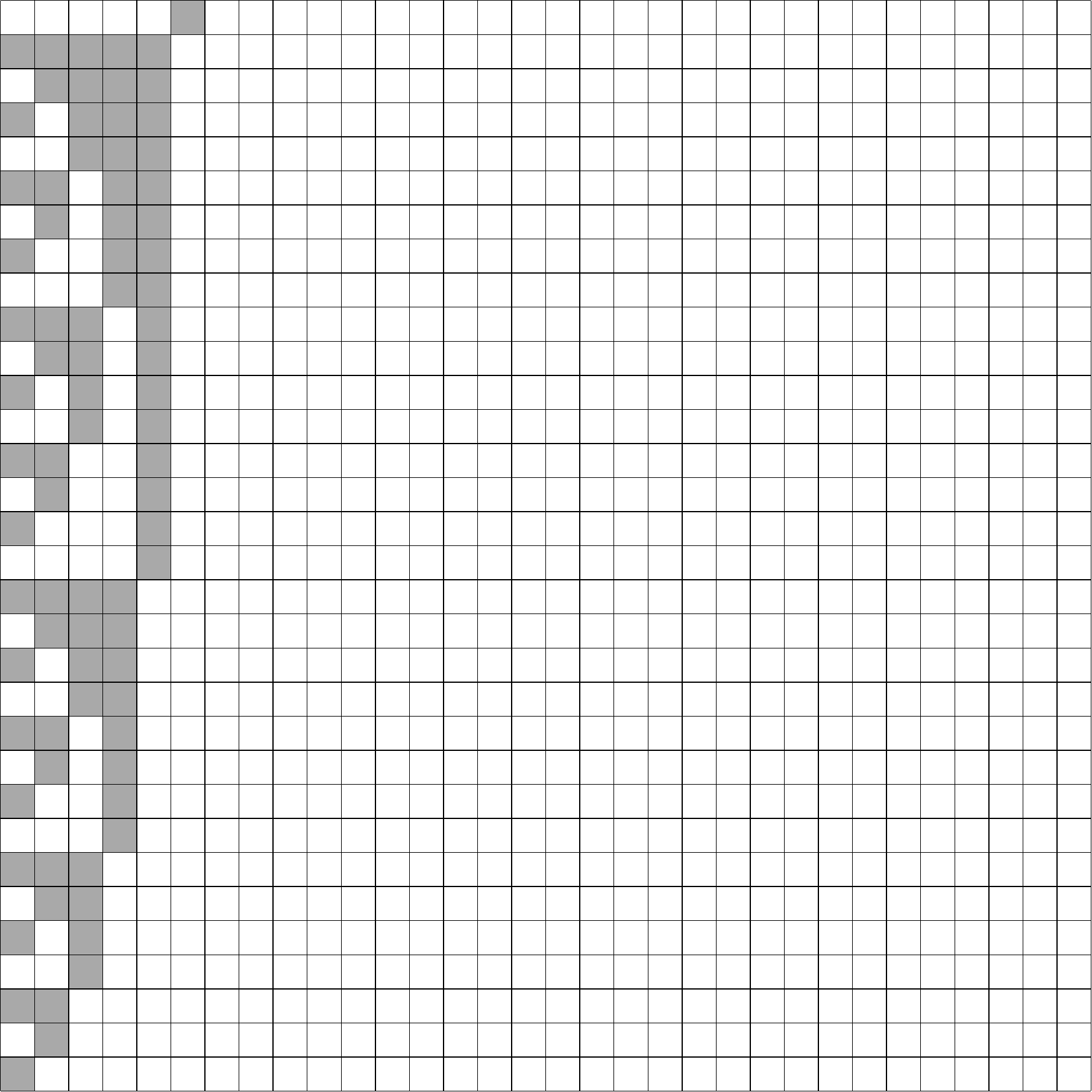}
    \label{fig:pat-bc32}
  }
  \subfigure[]{
    \includegraphics[width=0.3\textwidth]{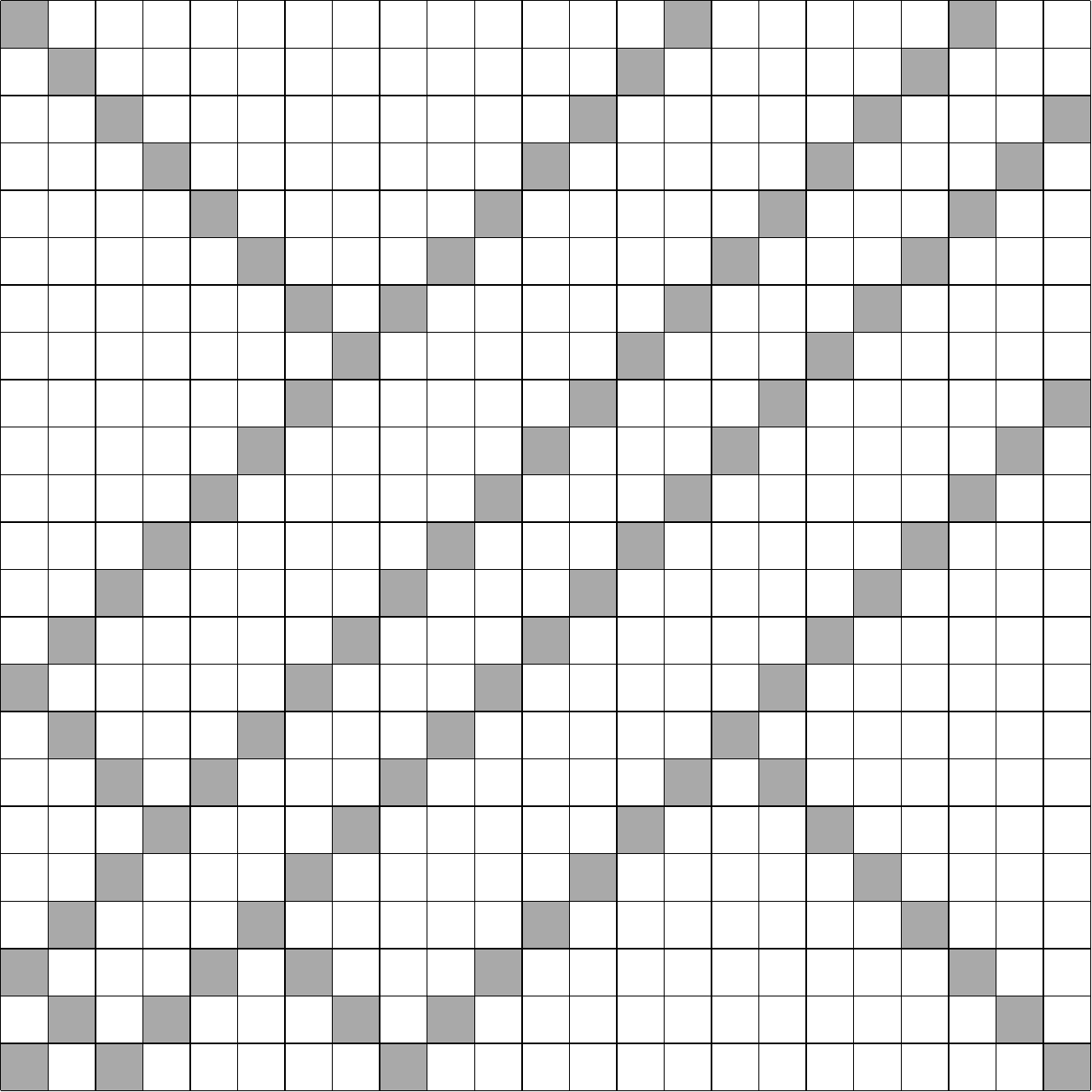}
    \label{fig:pat-ct23}
  }
  \subfigure[]{
    \includegraphics[width=0.98\textwidth]{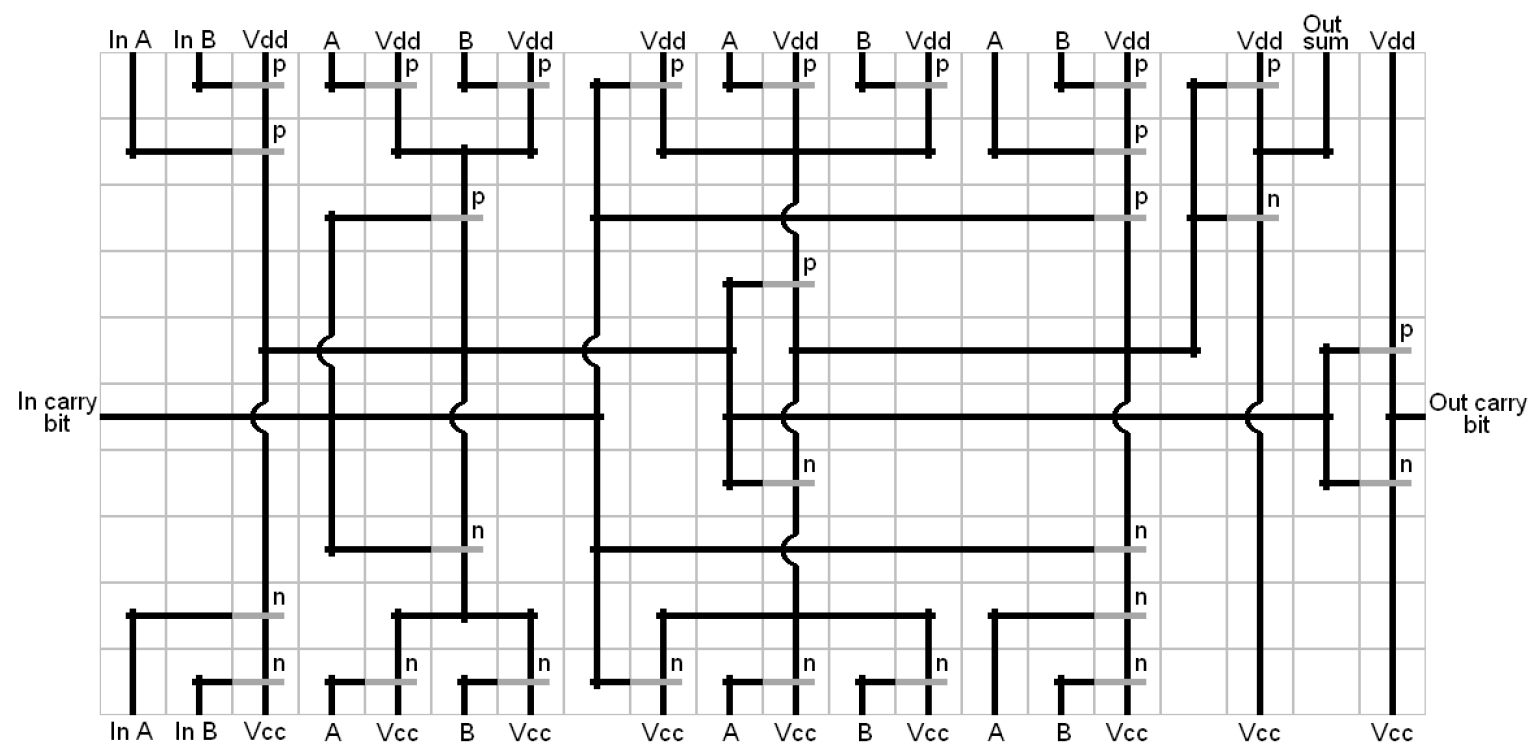}
    \label{fig:pat-fa}
  }
  \caption{\subref{fig:pat-srp32} The $32\times32$ Sierpinski triangle
    pattern. \subref{fig:pat-bc32} The $32\times32$ binary counter
    pattern.  \subref{fig:pat-ct23} The $23\times23$ ``tree'' pattern.
    \subref{fig:pat-fa} A CMOS full adder design that induces a
    15-colour $20\times10$ pattern.}
  \label{fig:patterns}
\end{figure}

Figure~\ref{fig:size-srp} presents the evolution of the ``current best
solution'' as a function of time for the \subref{fig:size-srp32-ct}
$32\times32$ and \subref{fig:size-srp64-ct} $64\times64$ Sierpinski
triangle patterns. To allow fair comparison, Figures~\ref{fig:size-srp32-tt}
and~\ref{fig:size-srp64-tt} present the same data with respect to the
total processing time taken by all the executions that run in parallel.
The experiments were repeated 21 times and the median of
the results is depicted. In 37\% of all the individual runs\footnote{In total there were $1\cdot21 + 2\cdot21 + 4\cdot21 + \cdots + 32\cdot21=1323$ runs for each input pattern.} conducted, the PS-H
algorithm was able to find the optimal 4-tile solution for the
$32\times32$ Sierpinski triangle pattern in less than 30 seconds. A similar
percentage for the $64\times64$ Sierpinski triangle pattern is 34\% in one
hour.  Remarkably, the algorithm performs only from 1030 to 1035 and
from 4102 to 4107 merge steps before arriving at the optimal solution
for the $32\times32$ and $64\times64$ patterns, respectively. In other
words, the search rarely needs to backtrack.  In contrast, the
smallest solutions found by the PS-BB algorithm have 42 tiles, reached
after $1.4\cdot10^6$ merge steps, and 95 tiles, reached after
$5.9\cdot10^6$ merge steps.

\begin{figure}[p]
  \centering
  \subfigure[]{
    \includegraphics[width=0.47\textwidth]{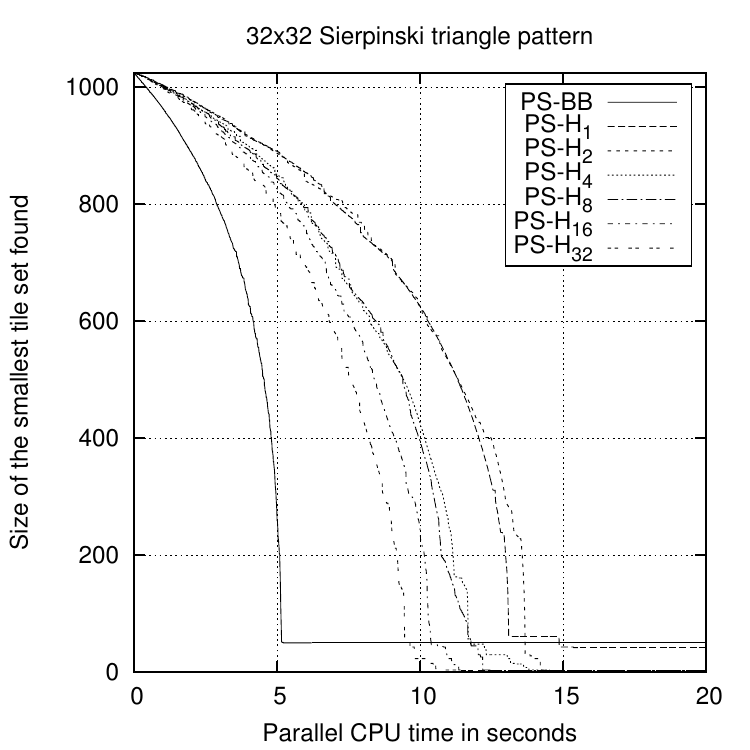}
    \label{fig:size-srp32-ct}
  }
  \subfigure[]{
    \includegraphics[width=0.47\textwidth]{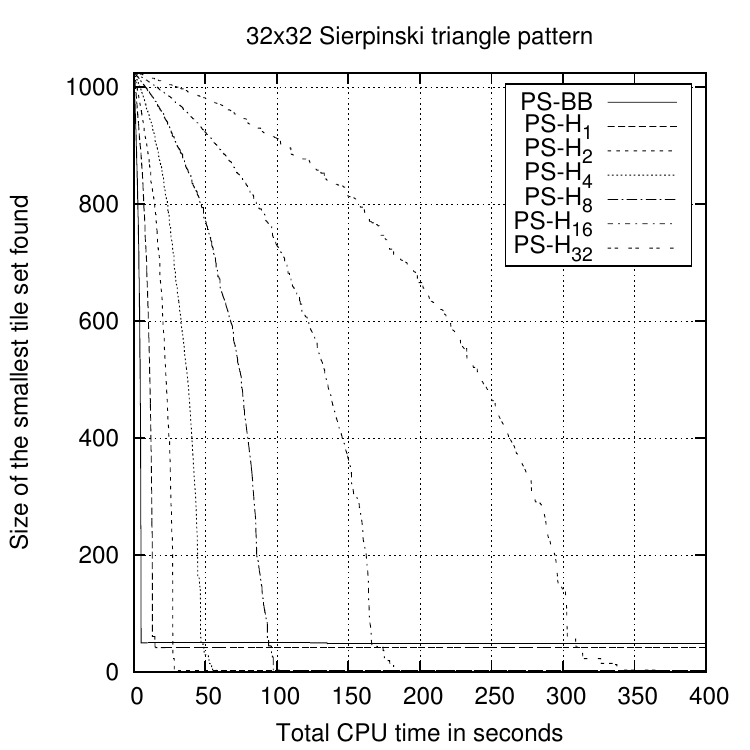}
    \label{fig:size-srp32-tt}
  }
  \subfigure[]{
    \includegraphics[width=0.47\textwidth]{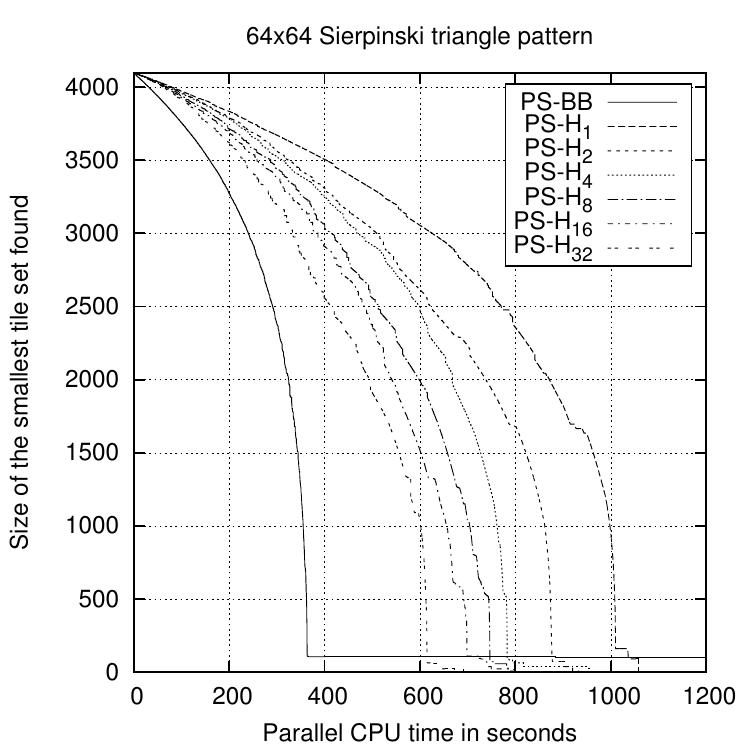}
    \label{fig:size-srp64-ct}
  }
  \subfigure[]{
    \includegraphics[width=0.47\textwidth]{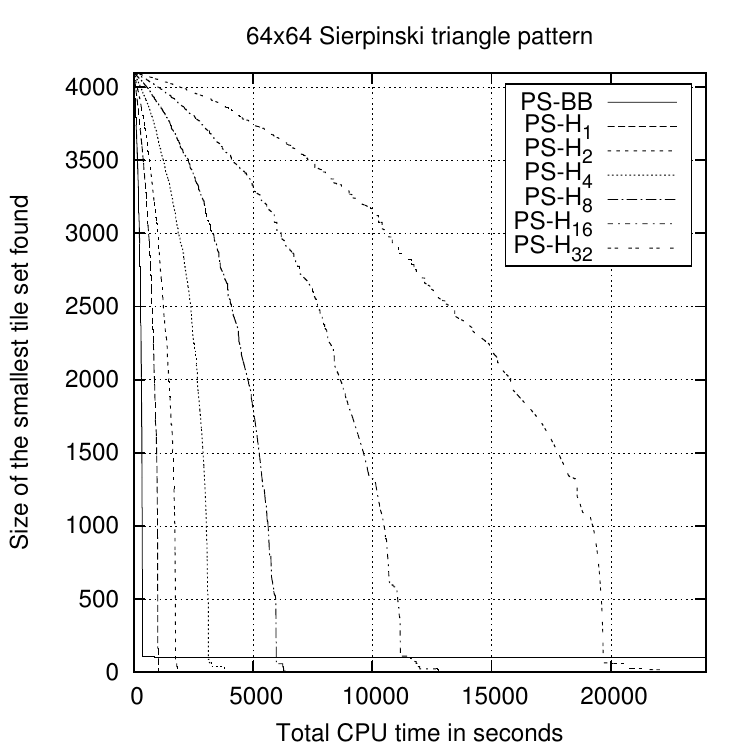}
    \label{fig:size-srp64-tt}
  }
  \caption{Evolution of the smallest tile set found for the $32\times32$ and
    $64\times64$ Sierpinski triangle patterns as a function of time.  The
    time axes measure \subref{fig:size-srp32-ct},
    \subref{fig:size-srp64-ct} CPU time and
    \subref{fig:size-srp32-tt}, \subref{fig:size-srp64-tt} CPU time
    multiplied by the number of parallel executions.}
  \label{fig:size-srp}
\end{figure}

\begin{figure}[p]
  \centering
  \subfigure[]{
    \includegraphics[width=0.47\textwidth]{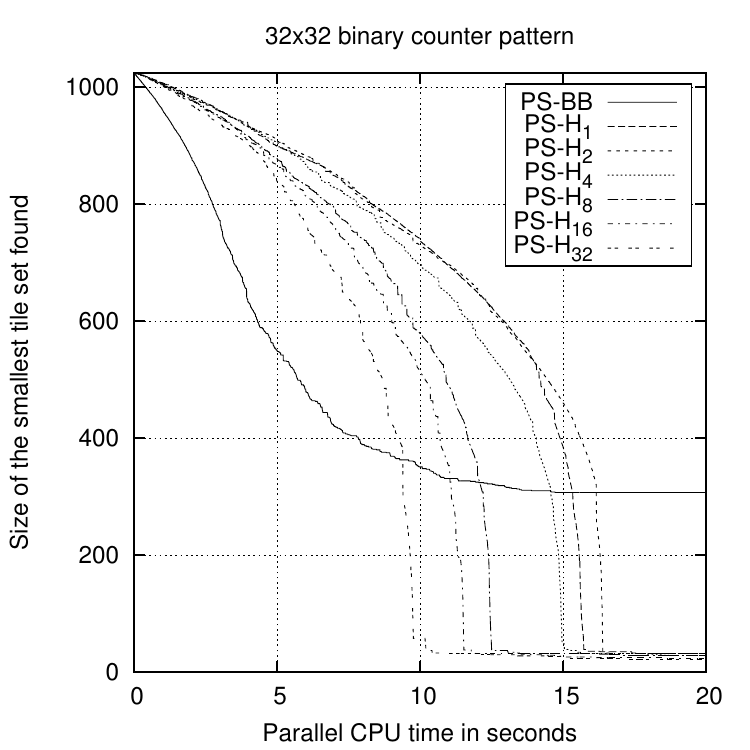}
    \label{fig:size-bc32-ct}
  }
  \subfigure[]{
    \includegraphics[width=0.47\textwidth]{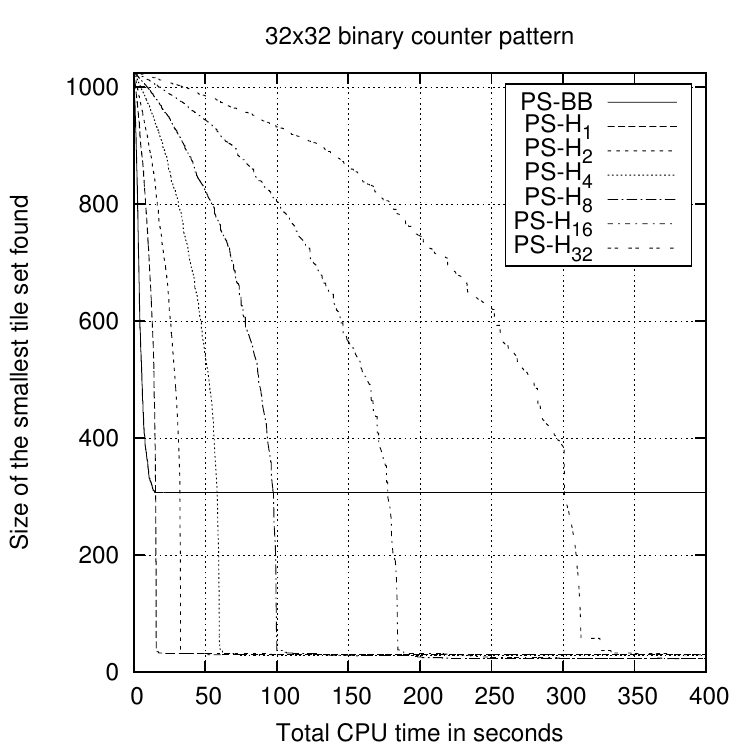}
    \label{fig:size-bc32-tt}
  }
  \subfigure[]{
    \includegraphics[width=0.47\textwidth]{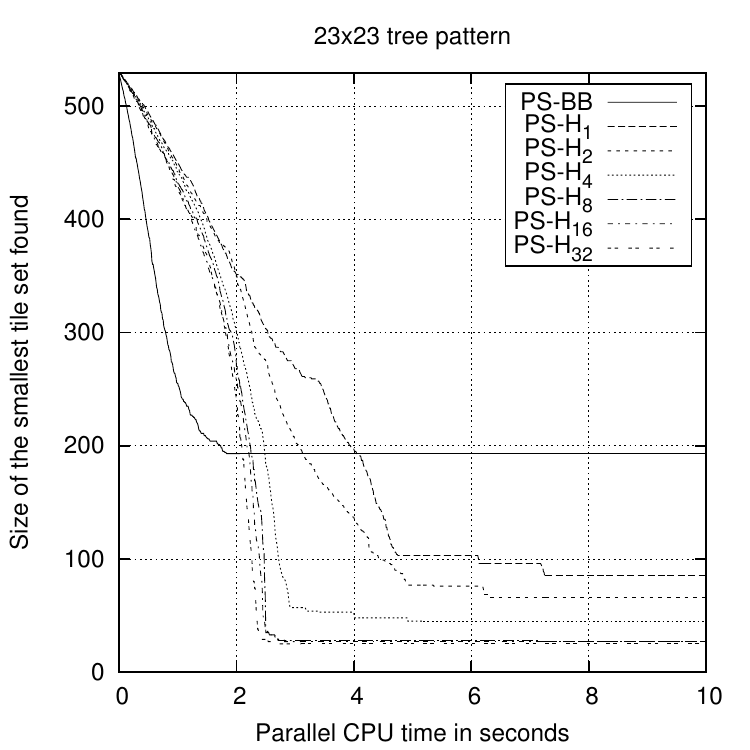}
    \label{fig:size-ct23-ct}
  }
  \subfigure[]{
    \includegraphics[width=0.47\textwidth]{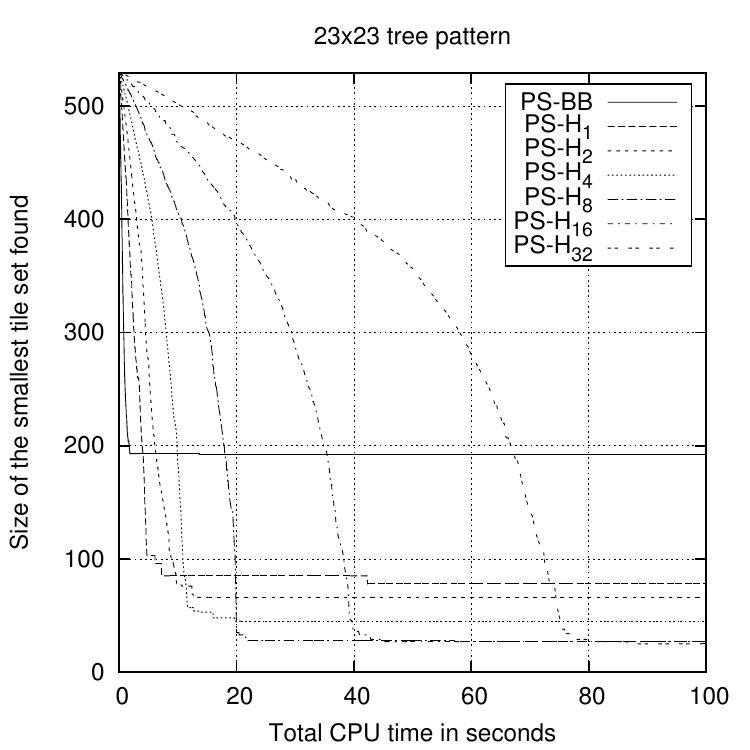}
    \label{fig:size-ct23-tt}
  }
  \caption{Evolution of the smallest tile set found for the $32\times32$
    binary counter and $23\times23$ tree patterns as a function of time.
    The time axes measure \subref{fig:size-bc32-ct},
    \subref{fig:size-ct23-ct} CPU time and \subref{fig:size-bc32-tt},
    \subref{fig:size-ct23-tt} CPU time multiplied by the number of
    parallel executions.}
  \label{fig:size-bcct}
\end{figure}

In Figure~\ref{fig:size-bcct} we present the corresponding results for
the $32\times32$ binary counter and $23\times23$ tree patterns. The
size of the smallest solutions found by the PS-H$_{32}$
algorithm were 20 (cf.\ 307 by PS-BB) and 25 (cf.\ 192 by PS-BB)
tiles, respectively. In the case of the tree pattern, the
parallelisation brings significant advantage over a single
run. Finally, Figures~\ref{fig:size-fa-ct}--\ref{fig:size-fa-tt} show
the results for the $20\times10$ 15-colour CMOS full adder pattern. In
this case, the improvement over the previous PS-BB algorithm is less
clear.  The PS-H$_{32}$ algorithm is able to find a solution of
58 tiles, whereas the PS-BB algorithm gives a solution of 69 tiles.

\begin{figure}[t!]
  \centering
  \subfigure[]{
    \includegraphics[width=0.47\textwidth]{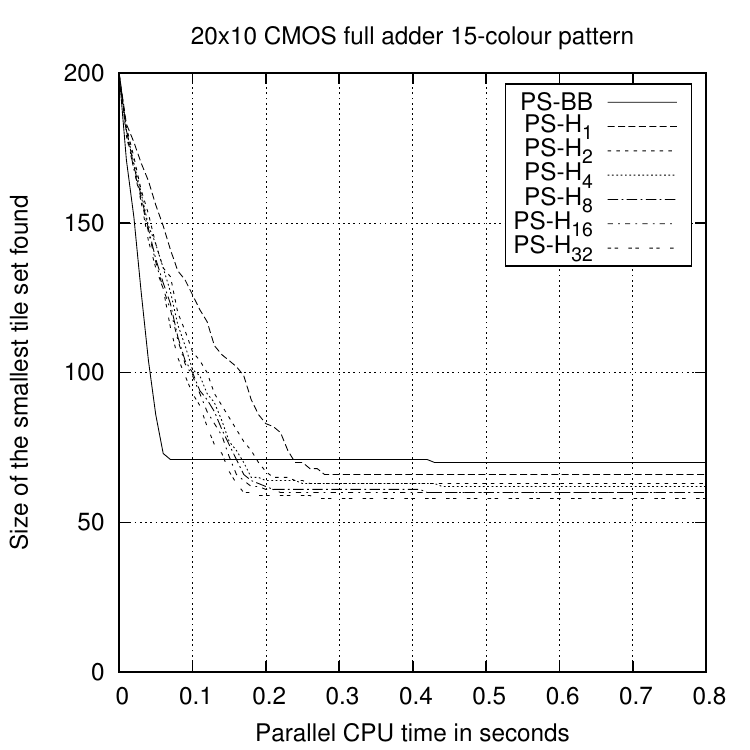}
    \label{fig:size-fa-ct}
  }
  \subfigure[]{
    \includegraphics[width=0.47\textwidth]{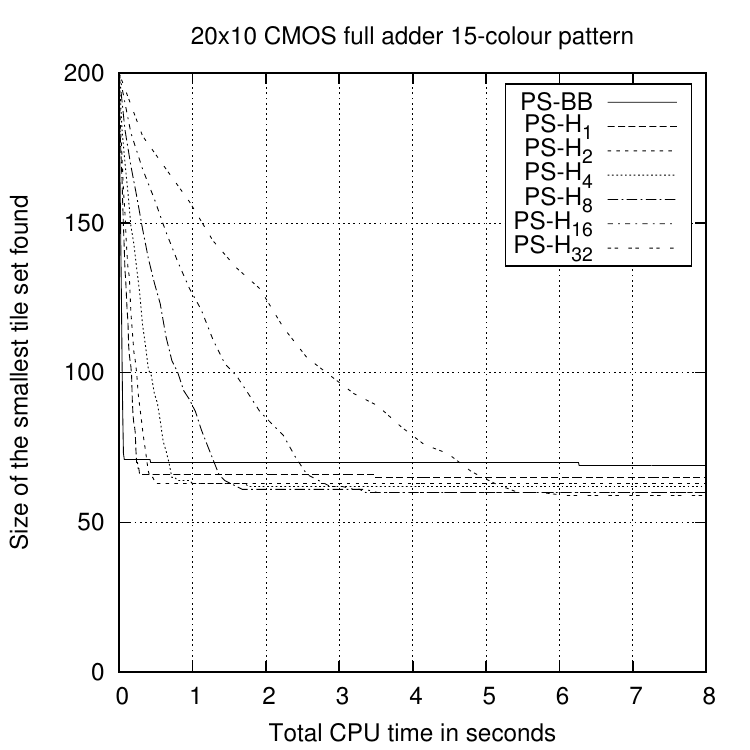}
    \label{fig:size-fa-tt}
  }
  \caption{Evolution of the
    smallest tile set found for the $20\times10$ full adder pattern as a
    function of time.  The time axes measure \subref{fig:size-fa-ct} CPU
    time and \subref{fig:size-fa-tt} CPU time multiplied by the
    number of parallel executions.
  }
  \label{fig:size-fa-sm}
\end{figure}

%% file: asp.tex
\subsection{An ASP Model for PATS}

Answer Set Programming (ASP)~\cite{Lifs08} is a declarative logic
programming paradigm for solving difficult combinatorial search
problems. In ASP, a problem is described as a logic program, and an
answer set solver is then used to compute stable models (answer sets)
of the logic program.  The ASP paradigm can be applied also to the
PATS problem. In the following we give a brief description on how to
transform the PATS problem to an ASP program using a modelling
language that is accepted by ASP grounders such as
\textsc{lparse}~\cite{Syrj98} or \textsc{gringo}~\cite{GKKS11}.

First, we define a constant for each position of the grid
$[m]\times[n]$, each colour, each available tile type and each
available glue type. After that, a number of choice rules are
introduced to associate a tile type with each position of the grid, a
glue type with each of the four sides of the tile types and a colour
with each of the tile types. Next, we use basic rules to make the
glues of every pair of adjacent tiles match and to make the tile
system deterministic, i.e.\ to ensure that every tile type has a
unique pair of glues on its W and S edges.
Finally, we compile the target pattern to a set of rules that associate
every position of the grid with the desired colour.

The above-described program is given to a grounder, which computes an
equivalent variable-free program. The variable-free program is forwarded to an
answer set solver, which then outputs a tile type for each position of
the grid, given that such a solution exists. We run the programs
repeatedly and increment the number of available tile and glue
types, until a solution is found.

\subsection{Results}
\label{sec:results-asp}

We used grounder \textsc{gringo}~3.0.5~\cite{GKKS11} and answer set solver
\textsc{clasp}~2.1.3~\cite{GKNS07} with default settings to run our experiments. A traditional
solver, \textsc{smodels}~\cite{NiSi97}, was also considered, but \textsc{clasp}
proved to be significantly faster in solving instances of the PATS problem.  We
consider two patterns having a minimal solution of 4 tiles: the Sierpinski
triangle and binary counter patterns. The programs were executed for
patterns of sizes $8\times8,16\times16,\ldots,256\times256$. We repeated the
experiments 21 times with different random seeds and the median running time is
presented in Figure~\ref{fig:sm-srp} for the Sierpinski triangle pattern and in
Figure~\ref{fig:sm-bc} for the binary counter pattern. The results include the
running time of both the grounder and the solver as well as all
the incremental steps needed until a solution is found. We were able to find
the minimal solution for both the $256\times256$ Sierpinski triangle pattern and
the $256\times256$ binary counter pattern in approximately 31 minutes of
(median) running time. The results were obtained on the same computing
cluster as the results in Section~\ref{sec:psh-res}.

Based on the above results, the ASP approach performs very well when
considering patterns with a small optimal solution. However, the
running time seems to increase dramatically with patterns that have a
larger optimal solution. Indeed, we were not able to find solutions for the
$23\times23$ tree pattern or the $20\times10$ CMOS full adder pattern using
the ASP approach.

\begin{figure}[t!]
  \centering

  \subfigure[]{
    \includegraphics[width=0.47\textwidth]{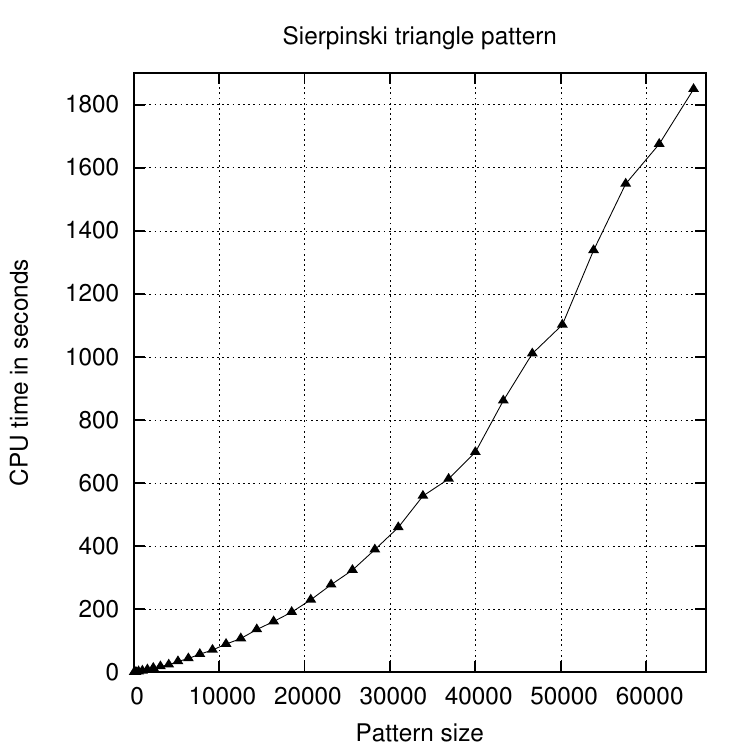}
    \label{fig:sm-srp}
  }
  \subfigure[]{
    \includegraphics[width=0.47\textwidth]{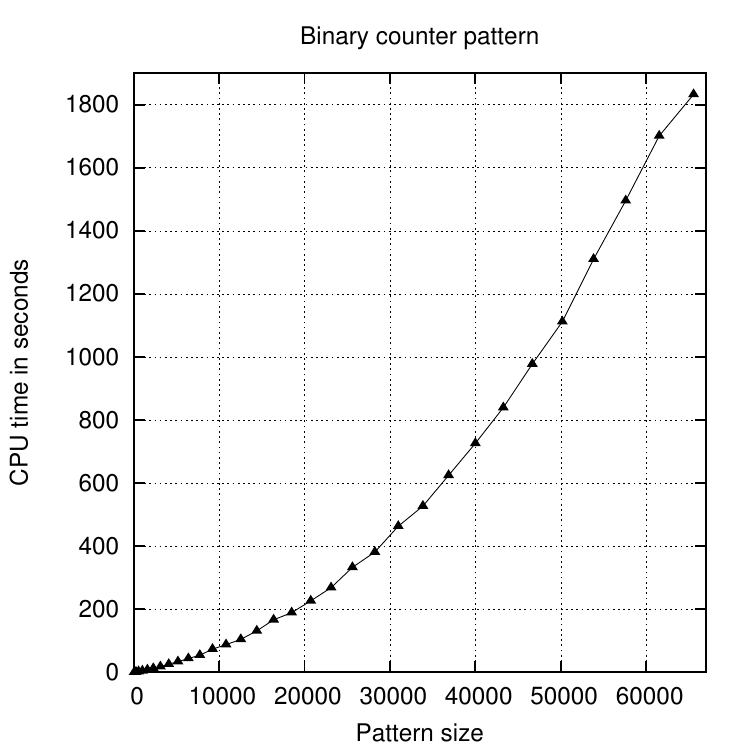}
    \label{fig:sm-bc}
  }
  \caption{Running time of \textsc{gringo} and \textsc{clasp} for the minimal solutions of the
    \subref{fig:sm-srp} Sierpinski triangle and \subref{fig:sm-bc} binary counter
    patterns as a function of pattern size.}
  \label{fig:size-sm}
\end{figure}

%% file: reliability.tex
\subsection{The Kinetic Tile Assembly Model}

In the following, we utilise the kinetic Tile Assembly Model (kTAM)
to assess the reliability of various tile sets generated by the
PS-BB and PS-H algorithms.  The kTAM was introduced by
Winfree~\cite{Winf98b} as a kinetic counterpart of the aTAM.
Several variants of the kTAM exist~\citep{FuMu09,ScWi09}.
However, the main elements are similar.

The kTAM simulates two types of reactions, each involving an assembly,
i.e.\ a crystal structure consisting of several merged tiles, and a
tile: \textit{association} of tiles to the assembly (forward reaction)
and \textit{dissociation} (reverse reaction), see e.g.\ Figure~\ref{fig:kbc}.\footnote{Note that
  interactions between two tiles, such as forming a new assembly, as
  well as interactions between two assemblies, are not taken into
  consideration in the initial model~\cite{Winf98b}.  However, they
  are studied in some of the later developed variants of the kTAM, see
  e.g.\ Schulman and Winfree~\cite{ScWi09}.} In the first type of reaction, any tile can
attach to the assembly at any position (up to the assumption that tile
alignment is preserved), even if only a weak bond is formed; the rate
of this reaction $r_f$ is proportional to the concentration of free
tiles in the solution. In the second type of reaction, any tile can
detach from the assembly with rate $r_{r,b}$, $b\in\{0,\ldots,4\}$,
which is exponentially correlated with the total strength of the bonds
between the tile and the assembly. Thus, tiles which are connected to
the assembly by fewer or weaker bonds, i.e.\ incorrect ``sticky end''
matches, are more prone to dissociation than those which are strongly
connected by several bonds (well paired sticky end sequences).

\begin{figure}[t]
  \centering
  \includegraphics[width=\textwidth]{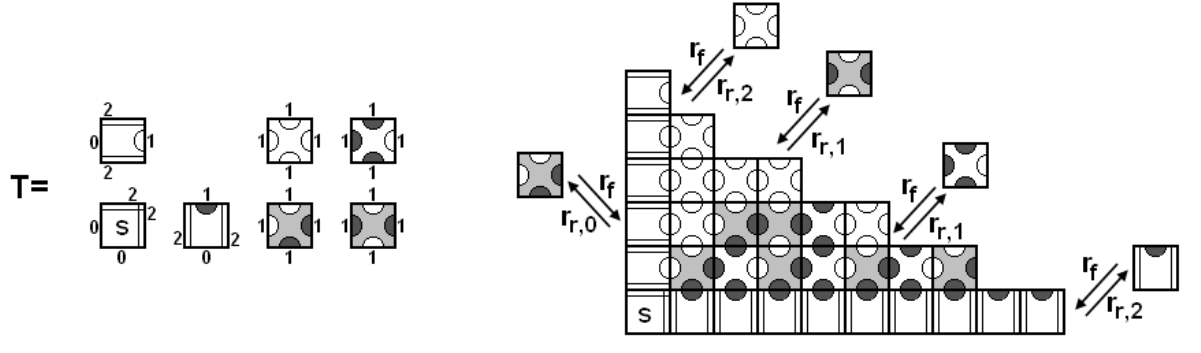}
  \caption{Possible association and dissociation reactions in the kinetic Tile Assembly Model. The rate of all the association reactions is identical; the rates of the dissociation reactions depend on the total strength of the bonds connecting a tile to the assembly.}
  \label{fig:kbc}
\end{figure}

In the following, we follow the notation of Winfree~\cite{Winf98b}.  For any
tile type $t$, the rate constant $r_f$ of the association (forward
reaction) of $t$ to an existing assembly is given by
\[
r_f=k_f[t],\quad (\mbox{in /sec})
\]
where $[t]$ is the concentration in solution of free tiles of type $t$
and $k_f$ is a temperature dependent parameter. In the case of DNA
double-crossover (DX) tiles, this parameter is given by the formula
\[
k_f=A_f e^{-E_f/ RT},
\]
where $A_f=5 \cdot 10^8\ \mathrm{/M/sec}$,
$E_f=4000\ \mathrm{cal/mol}$, $R=2\ \mathrm{cal/mol/K}$, and $T$ is the
temperature (in $\mathrm{K}$).

In the case of dissociation (reverse reaction), for a tile which is
connected to the assembly by a total bond strength $b$, the rate
constant $r_{r,b}$ is given by the formula
\[
r_{r,b}=k_fe^{\Delta G_b^o/RT},
\]
where $\Delta G_b^o$ is the standard free energy needed to
break $b$ bonds. In the case of DX tiles, as the glues of the
tiles are implemented using 5-base long single-stranded DNA molecules,
$\Delta G_b^o$ can be estimated using the nearest-neighbour
model~\cite{SaAS96} as
\[
\Delta G_b^o=e^{5b\left(11-\frac{4000\ \mathrm{K}}{T}\right)+3}\
\mathrm{cal/mol}.
\]
Moreover, $b$ can range with integer values from $0$ to $4$,
corresponding to the cases when the tile is totally erroneously placed
in the assembly (no bond connects it to the crystal) and when the tile
is fully integrated into the assembly (all its four sticky ends are
correctly matched), respectively.

In order to easily represent and scale the system, the free parameters
involved in the formulas of the rate constants $r_{f}$ and $r_{r,b}$
are re-distributed into just two dimensionless parameters, $G_{mc}$
and $G_{se}$, where the first is dependent on the initial tile
concentration and the second is dependent on the assembly
temperature:
\[
r_{f}=\hat{k}_f e^{-G_{mc}}, \qquad r_{r,b}=\hat{k}_f e^{-bG_{se}},
\]
where, in the case of DX tiles, $\hat{k}_f=e^3k_f$ is adjusted in
order to take into consideration possible entropic factors, such as
orientation or location of tiles. The previous parameter
re-distribution is made possible as a result of the assumption made in
the initial kTAM~\cite{Winf98b} that all tile types are provided into
the solution in similar concentrations, and that the consumption in
time of the free monomers is negligible compared to the initial
concentration.

\subsection{Computing the Reliability of a Tile Set}

By choosing appropriate physical conditions, the probability of errors
in the assembly process can be made arbitrarily low, at the cost of
reducing the assembly rate~\cite{Winf98b}. However, we would like to
be able to compare the error probability of different tile sets
producing the same finite pattern, under the same physical conditions.
Given the amount of time the assembly process is allowed to take, we
define the \emph{reliability of a tile set} to be the probability that
the assembly process of the tile system in question completes without
any incorrect tiles being present in the terminal configuration. In
the following, we present a method for computing the reliability of a
tile set, based on Winfree's analysis of the kTAM~\cite{Winf98b}, and
the notion of \emph{kinetic trapping} introduced within.

We call the W and S edges of a tile its \emph{input edges}. First, we
derive the probability of the correct tile being frozen at a
particular site under the condition that the site already has correct
tiles on its input edges. Let $M^1_{i,j}$ and $M^2_{i,j}$ be the
number of tile types having one mismatching and two mismatching input
glues, respectively, between them and the correct tile type for site
$(i,j)\in[m]\times[n]$. Now, for a deterministic tile set $T$, the
total number of tiles is $|T| = 1+M^1_{i,j}+M^2_{i,j}$ for any
$(i,j)\in[m]\times[n]$. Given that a site has the correct tiles on its
input edges, a tile is correct for that site if and only if it has two
matches on its input edges.

In what follows, we assume that correct tiles are attached at sites
$(i-1,j)$ and $(i,j-1)$. The model for kinetic trapping~\cite{Winf98b}
gives four distinct cases in the situation preceding the site $(i,j)$
being frozen by further growth. To each of these cases we can
associate an ``off-rate'' for the system to exit its current state:
(E) An empty site, with off-rate $|T|r_f$.  (C) The correct tile,
with off-rate $r_{r,2}$. (A) A tile with one match, with off-rate
$r_{r,1}$. (I) A tile with no matches, with off-rate
$r_{r,0}$. Additionally, we have two sink states FC and FI, which
represent frozen correct and frozen incorrect tiles, respectively. The
rate of a site being frozen is equal to the rate of growth $r^* =
r_f-r_{r,2}$. Figure~\ref{figFk} describes the dynamics of the system.
Let $p_S(t)$ denote the probability of the site being in state $S$
after $t$ seconds for all $S\in\{\mathrm{E, C, A, I, FC, FI}\}$. To
compute the frozen distribution, we write the rate equations for the model of
kinetic trapping from Figure~\ref{figFk} as follows:\footnote{The notation
$\mathbf{\dot{p}}(x)$ is used to denote the derivative of $\mathbf{p}$ with
respect to time.}
\begin{figure}[t!]
  \centering
  \includegraphics{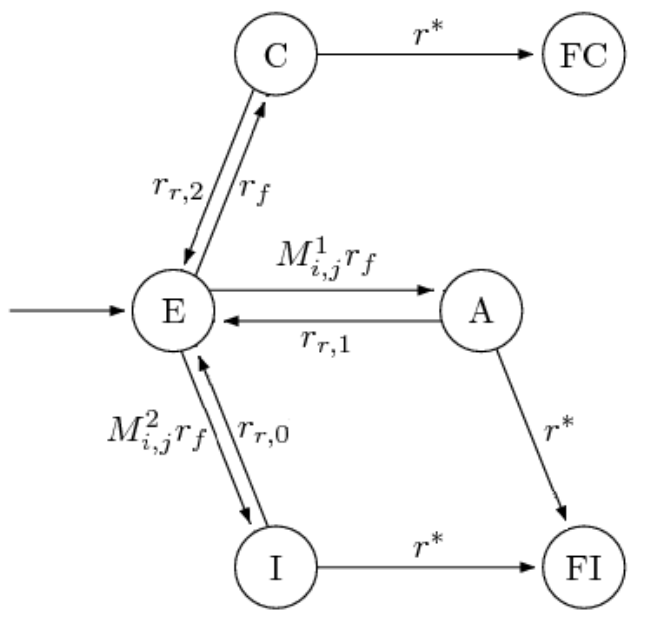}
  \caption{The dynamics of the kinetic trapping model.}
  \label{figFk}
\end{figure}
\[
  M\mathbf{p}(t) \coloneqq
  \begin{bmatrix}
    -|T|r_f     & r_{r,2}      & r_{r,1}      & r_{r,0}      & 0\quad& 0 \\
    r_f         & -r_{r,2}-r^* & 0           & 0           & 0\quad& 0 \\
    M^1_{i,j}r_f & 0           & -r_{r,1}-r^* & 0           & 0\quad& 0 \\
    M^2_{i,j}r_f & 0           & 0           & -r_{r,0}-r^* & 0\quad& 0 \\
    0           & r^*         & 0           & 0           & 0\quad& 0 \\
    0           & 0           & r^*         & r^*         & 0\quad& 0
  \end{bmatrix}
  \begin{bmatrix}
    p_\mathrm{E}(t) \\
    p_\mathrm{C}(t) \\
    p_\mathrm{A}(t) \\
    p_\mathrm{I}(t) \\
    p_{\mathrm{FC}}(t) \\
    p_{\mathrm{FI}}(t)
  \end{bmatrix}
  = \mathbf{\dot{p}}(t),
\]
where $\mathbf{p}(0) = \begin{bmatrix} 1 & 0 & 0 & 0 & 0 & 0 \end{bmatrix}^T$.
To compute the steady-state probability of the site being frozen with
the correct tile, i.e.\ $p_{\mathrm{FC}}(\infty)$, we make use of the
steady state of the related flow problem~\cite{Winf98b}:\footnote{By
  the definition of the kinetic trapping model~\cite{Winf98b}, it is
  assumed that a unit amount of material is supplied into state
  $\mathrm{E}$ of the system at any time point.}
\[
  M\mathbf{p}(\infty) =
  \begin{bmatrix}
    1 & 0 & 0 & 0 & p_{\mathrm{FC}}(\infty) & p_{\mathrm{FI}}(\infty)
  \end{bmatrix}^T
  = \mathbf{\dot{p}}(\infty),
\]
which gives us a system of linear equations. This system has a single
solution, namely
\[
p_{\mathrm{FC}}(\infty) =
\frac{\frac{1}{r^*+r_{r,2}}}{\frac{1}{r^*+r_{r,2}}+
  \frac{M^1_{i,j}}{r^*+r_{r,1}}+\frac{M^2_{i,j}}{r^*+r_{r,0}}} =
\Pr(C_{i,j}\,|\,C_{i-1,j} \cap C_{i,j-1}),
\]
where $C_{i,j}$ denotes the event of the correct tile being frozen at
site $(i,j)$.

The assembly process can be thought of as a sequence $(a_1, a_2, \ldots, a_N)$
of tile addition steps where $a_k = (i_k,j_k)$, $k =
1,2,\ldots,N$, denotes a tile being frozen at site $(i_k,j_k)$. Due to
the fact that the assembly process of the tile systems considered here
proceeds uniformly from south-west to north-east, we have that
$\{(i_k-1,j_k), (i_k,j_k-1)\} \subseteq \{a_1, a_2, \ldots, a_{k-1}\}$
for all $a_k = (i_k,j_k)$. We assume that tiles elsewhere in the
configuration do not affect the probability. Now we can compute the
probability of a finite-size pattern of size $N$ assembling without
any errors, i.e.\ the reliability of that pattern:
\begin{align*}
  \Pr(\mbox{correct pattern})
  & = \Pr(C_{a_1} \cap C_{a_2} \cap \dotsb \cap C_{a_N}) \\
  & = \Pr(C_{a_1}) \Pr(C_{a_2}\,|\,C_{a_1}) \dotsm \Pr(C_{a_N}\,|\,C_{a_1}
   \cap C_{a_2} \cap \dotsb \cap C_{a_{N-1}}) \\
  & = \prod_{i,j} \Pr(C_{i,j}\,|\,C_{i-1,j} \cap C_{i,j-1}).
\end{align*}

We have computed the probability in terms of $G_{mc}$ and
$G_{se}$. Given the desired assembly rate, we want to minimise the
error probability by choosing values for $G_{mc}$ and $G_{se}$
appropriately. If the assembly process is allowed to take $t$ seconds,
the needed assembly rate for an $m \times n$ pattern is approximately
$r^* = \frac{\sqrt{m^2+n^2}}{t}$.
In order to simplify the computations, we use the approximation
\[
\Pr(C_{i,j}\,|\,C_{i-1,j} \cap C_{i,j-1}) =
\frac{\frac{1}{r^*+r_{r,2}}}{\frac{1}{r^*+r_{r,2}}+
\frac{M^1_{i,j}}{r^*+r_{r,1}}+\frac{M^2_{i,j}}{r^*+r_{r,0}}} \approx
\frac{1}{1+M^1_{i,j}\frac{r^*+r_{r,2}}{r^*+r_{r,1}}}.
\]
For small error probability and $2G_{se} > G_{mc} > G_{se}$,
\[
\Pr(\neg C_{i,j}\,|\,C_{i-1,j} \cap C_{i,j-1}) \approx
M^1_{i,j}\frac{r^*+r_{r,2}}{r^*+r_{r,1}} \approx
M^1_{i,j}e^{-(G_{mc}-G_{se})} \eqqcolon
M^1_{i,j}e^{-\triangle G}.
\]
From
\[
r^* = r_f - r_{r,2} = \hat{k}_f(e^{-G_{mc}}-e^{-2G_{se}})
\]
we can derive
\[
G_{se} = -\frac{1}{2}\log(e^{-G_{mc}}-\frac{r^*}{\hat{k}_f}).
\]
Now we can write $\triangle G$ as a function of $G_{mc}$:
\[
\triangle G(G_{mc}) = G_{mc} - G_{se} =
G_{mc} + \frac{1}{2}\log(e^{-G_{mc}}-\frac{r^*}{\hat{k}_f}).
\]
We find the maximum of $\triangle G$, and thus the minimal error probability, by
differentiation:
\[
G_{mc} = -\log(2\frac{r^*}{\hat{k}_f}).
\]
Thus, if the assembly time is $t$ seconds, the maximal reliability is achieved
at
\[
G_{mc} = -\log(2\frac{\sqrt{m^2+n^2}}{t\hat{k}_f}),\qquad
G_{se} = -\frac{1}{2}\log(\frac{\sqrt{m^2+n^2}}{t\hat{k}_f}).
\]

\subsection{Results}\label{sec:results-reliability}

In this section, we present results on computing the reliability of
tile sets using the method given above. We assume that the assembly process
takes place in room temperature (298 K). As a result, we use the value
$k_f=A_fe^{-E_f/RT} \approx 6 \cdot 10^5\mbox{ /M/sec}$ for the
forward reaction rate.

Figure~\ref{fig:rel-srp4t} shows the reliability of the 4-tile
solution to the Sierpinski triangle pattern as a function of pattern
size, using five distinct assembly times. As is to be expected, the
longer the assembly time, the better the reliability.

\begin{figure}[p]
  \centering
  \subfigure[]{
    \includegraphics[width=0.47\textwidth]{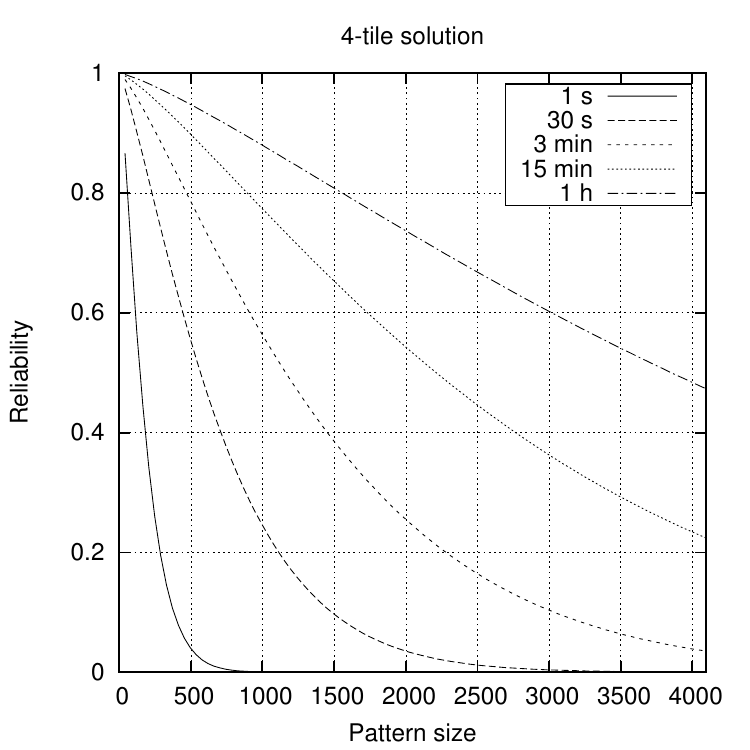}
    \label{fig:rel-srp4t}
  }
  \subfigure[]{
    \includegraphics[width=0.47\textwidth]{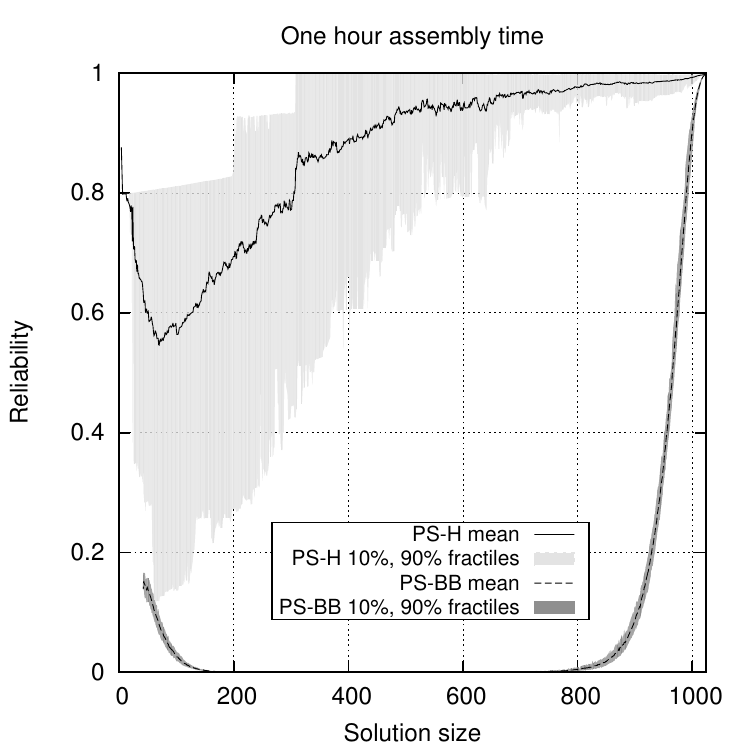}
    \label{fig:rel-srp32-1h}
  }
  \subfigure[]{
    \includegraphics[width=0.47\textwidth]{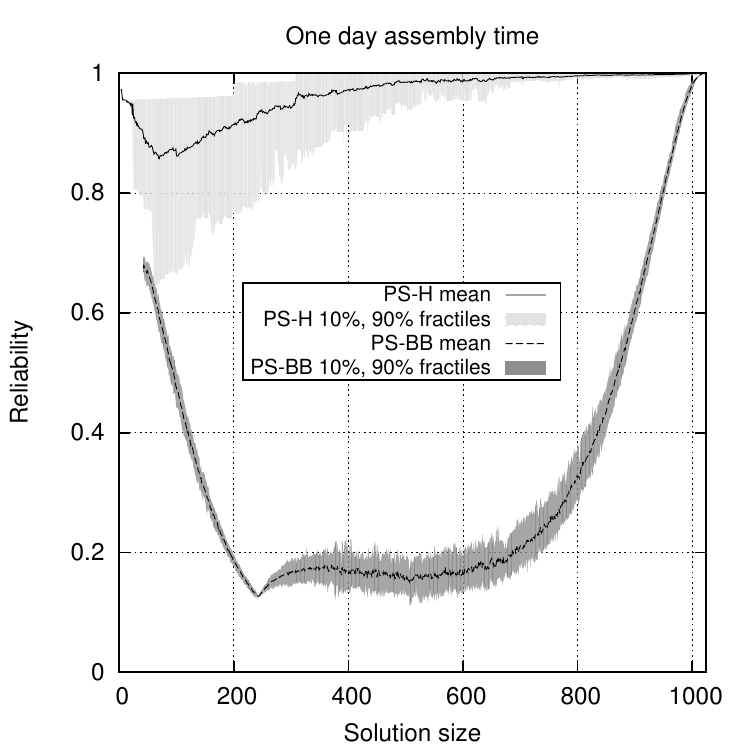}
    \label{fig:rel-srp32-1d}
  }
  \subfigure[]{
    \includegraphics[width=0.47\textwidth]{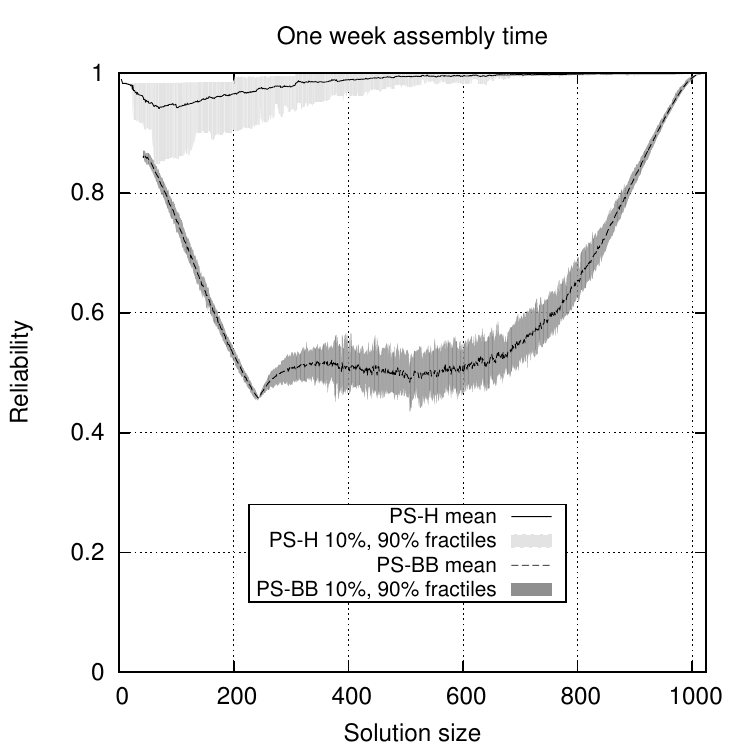}
    \label{fig:rel-srp32-1w}
  }
  \caption{\subref{fig:rel-srp4t} The reliability of the minimal tile set as a
  function of pattern size for the Sierpinski triangle pattern, using several
  different assembly times.
  \subref{fig:rel-srp32-1h}--\subref{fig:rel-srp32-1w} The reliability of
  solutions for the $32\times32$ Sierpinski triangle pattern found by the PS-H and
  PS-BB algorithms, allowing assembly time of one hour, one day and one
  week.}
  \label{fig:reliability}
\end{figure}

We also applied the method for computing the reliability to tile sets
found by the partition-search algorithms. Our results show that the
heuristics used in the PS-H algorithm improve not only the size of the
tile sets found, but also the reliability of those tile sets. This can
be easily understood by considering the following: The reliability of
a tile set is largely determined by the number of tile types that have
the same glue as some other tile type on either one of their input
edges. Since the PS-H algorithm prefers merging class pairs with
common glues, it reduces the number of such tile types effectively.

Figures \ref{fig:rel-srp32-1h}--\ref{fig:rel-srp32-1w} present the
reliability of tile sets found by the PS-H and PS-BB algorithms for
the $32\times32$ Sierpinski triangle pattern, with assembly times of
one hour, one day (24 hours) and one week.  The runs were repeated 100
times; the mean reliability of each tile set size as well as the 10th
and 90th percentiles are shown.

As for reliability, we expect a large set of runs of the PS-BB
algorithm to produce a somewhat decent sample of all the possible tile
sets for a pattern.  Based on this, large and small tile sets seem to
have a high reliability while medium-size tile sets are clearly less
reliable on average. This observation reduces the problem of finding
reliable tile sets back to the problem of finding small tile sets.
However, it is important to note that artefacts of the algorithm may
have an effect on the exact reliability of the tile sets found.

%% file: conclusions.tex
We have investigated several algorithmic approaches towards an
efficient solution to the PATS problem, i.e.\ the task of finding
minimal tile sets which would self-assemble into a given $k$-coloured
pattern starting from a bordering seed structure.

Our first algorithm is an exhaustive branch-and-bound method (PS-BB)
which makes use of a search tree in the lattice of grid
partitions. Given enough time, the algorithm finds a provably minimal tile set
for any pattern. Numerical experiments indicate that the
PS-BB algorithm is able to find minimal tile sets for randomly
generated binary patterns of sizes up to $6 \times 6$ tiles.  However,
for larger patterns, the search space becomes too large for a complete
exploration, even with the efficient pruning methods applied by the
algorithm.

In a second approach, we addressed the relaxed objective of generating small
but not necessarily minimal tile sets.  Here our PS-H algorithm applies
heuristic rules for optimising the order in which the search space of
pattern-consistent tile sets is explored. Experimental results show that for
most patterns, the PS-H algorithm is indeed able to find significantly smaller
solutions than the PS-BB algorithm, in a reasonable amount of time.

In a third direction, we also considered solving the PATS problem
using logic programming techniques, specifically the Answer Set
Programming (ASP) method. For patterns having small optimal solutions,
our chosen ASP solver is mostly very successful in discovering these
solutions; however the running time of the solver seems to increase
rapidly with the size of the minimum solution.

On a supporting topic, we used the kinetic Tile Assembly Model to assess the
reliability of various tile sets generated by the PS-BB and PS-H algorithms,
i.e.\ their probability of assembling the desired target pattern in an
error-free manner. In comparison to the PS-BB approach, we find that the
heuristics used in the PS-H algorithm improve also the reliability of tile sets
found. In addition, we observed that large and small tile sets seem to have a
high reliability, while medium-size tile sets are clearly less reliable on
average.

One research question still open is the \textsf{NP}-hardness of the PATS
problem restricted to 2-colour patterns. As for new solving methods, further
work could include developing polynomial-time approximation algorithms. The
declarative approach could possibly be applied to instances with larger optimal
solutions by developing a more efficient ASP or propositional satisfiability
encoding.

%% file: acks.tex
We thank the anonymous reviewers for their helpful feedback. The
demanding numerical computations were performed on the Triton
computing cluster provided by the Aalto University Science-IT
programme.

%% file: cnfet.tex
Recent years have witnessed a burst of experimental activity
concerning algorithmic self-assembly of nanostructures, motivated at
least in part by the potential of this approach as a radically new
manufacturing technology. One of the presently most reliable
self-assembling, programmable nanostructure architectures is DNA
origami \cite{Roth06}. Several authors have announced the formation of
DNA origami tiles, capable of further assembly into larger, fully
addressable, 1D and 2D scaffolds \cite{ESKH10,KSML10,LZWS11}. Such
scaffolds make possible the construction of highly complex structures
on top of them \cite{KuLT09}, prospectively including nanocircuits.
In Czeizler et al.~\cite{CzLO11}, we proposed a generic framework for the design of
Carbon Nanotube Field Effect Transistor (CNFET) circuits. The elements
of these circuits are Carbon Nanotube Field Effect Transistors and
Carbon Nanotube Wires. They are placed on top of different DNA origami
tiles which self-assemble into any desired circuit.

Single-wall carbon nanotubes (CNs) can be fabricated as either
metallic (m) or semiconducting (s). A cross-junction between an m-type
and an s-type CN generates a structure with field effect transistor
(FET) behaviour \cite{DJCP04,LSLP06}. In this way, both p-type and
n-type FETs are realisable (a p-type FET is ON when input is ``0'',
while an n-type FET is ON when input is ``1''). Moreover, experimental
implementations have been provided, affixing these structures on top
of DNA origami \cite{EKKT11,MHBB10}.

\begin{figure}[ht]
  \centering
\includegraphics[width=11cm]{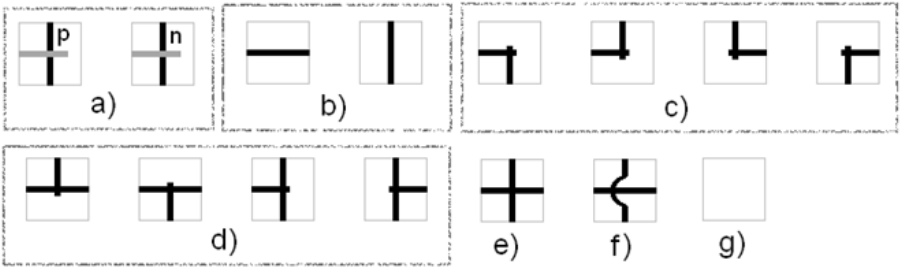}
  \caption{The 14 tile types and the blank tile, out of which any
CNFET circuit can be assembled: (a) p-type and n-type CNFETs, (b) straight CNWs, (c)~corner CNWs, (d)--(e) 3-way and 4-way CNW junctions, (f) crossing but non-interacting CNWs and (g) blank tile.}
  \label{fig:tiles}
\end{figure}

Based on the above experimental results, we provided in Czeizler et al.~\cite{CzLO11}
a ``universal'' set of 14 functionalised DNA origami tiles, such that,
with a proper selection of ``glues'' on the tiles, any desired CNFET
circuit can be self-assembled from this basis. These tile types are
presented in Figure~\ref{fig:tiles} (the marks on the tiles indicate
the arrangements of the CNs affixed on the respective DNA origami):
(a) p-type and n-type CNFETs, (b) straight (horizontal or vertical) CN
wires (CNWs), (c) corner CNWs, (d)--(e) 3-way and 4-way junction CNWs
and (f) crossing but non-interacting CNWs. Additionally, when
analysing fault tolerant architectures, it is convenient to introduce
also (g) a blank tile. In order to design a particular nanocircuit,
one first prepares the transistor circuit design using the 14 basis
tiles indicated.  Then, an optimal number of glues for these tiles is
computed and finally, appropriate ``sticky end'' sequences for
implementing the glues are designed for the DNA origami tiles. In
Figure~\ref{fig:cnfexp} we present the designs for a CMOS inverter,
NAND gate and full adder.

\begin{figure}[ht]
  \centering
  \includegraphics[width=13cm]{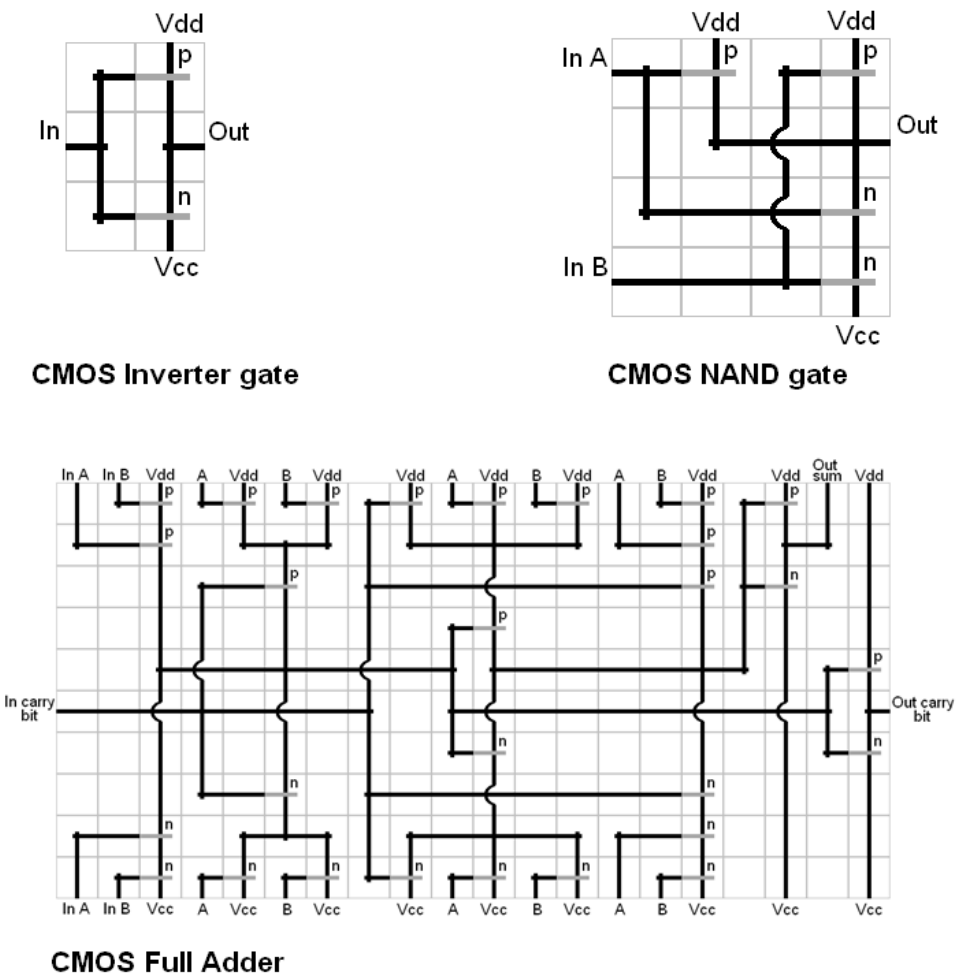}
  \caption{Examples of CNFET circuit design: an inverter gate, a NAND gate and a full adder.}
  \label{fig:cnfexp}
\end{figure}

Some of the advantages of this approach are that it decouples the
self-assembly aspects of the manufacturing process from the transistor
circuit design and that it allows for a structured and clear circuit
design. Moreover, it also supports efficient high-level analysis of
the purported circuits, both by computer simulations and by analytical
means. For instance, all assembly errors can at this level be treated
as tiling errors, leading to a transparent design discipline for
fault-tolerant architectures.